\documentclass[11pt]{article}
\usepackage[noend]{algpseudocode}
\usepackage{amsthm}
\usepackage{amsmath}
\usepackage{amssymb}
\usepackage{fullpage}
\usepackage{times}
\usepackage{paralist}
\usepackage[usenames]{color}
\usepackage{enumerate}
\usepackage{tabularx}
\usepackage[toc,page]{appendix}
\usepackage{graphicx}
\usepackage{mathtools}
\usepackage{algorithm}
\usepackage{algorithm,graphicx,subfig,multirow,url,tikz,algpseudocode}
\usepackage{hyperref}

\renewcommand{\phi}{\varphi}
\newcommand{\set}[1]{\left\{#1\right\}}

\newcommand{\coloneq}{:=}
\newcommand{\st}{\medspace | \medspace}
\newcommand{\hide}[1]{ }

\DeclareMathOperator*{\E}{E}
\renewcommand{\phi}{\varphi}

\newcommand{\eps}{\epsilon}

\newcommand{\prob}[1]{\text{\textsc{#1}}}
\newcommand{\given}{\medspace | \medspace}
\newcommand{\rv}[1]{ \boldsymbol{ #1} }
\DeclareMathOperator*{\MI}{I}
\DeclareMathOperator*{\KLDiv}{D}
\newcommand{\Div}[2]{\KLDiv\left( #1 \parallel #2 \right)}

\newcommand{\fnstyle}[1]{\mathsf{#1}}

\theoremstyle{plain}
\newtheorem{theorem}{Theorem}[section]

\newtheorem{lemma}[theorem]{Lemma}
\newtheorem{observation}[theorem]{Observation}
\newtheorem{corollary}[theorem]{Corollary}

\newtheorem{definition}{Definition}

\renewcommand{\include}{\input}

\title{On The Multiparty Communication Complexity of Testing Triangle-Freeness}
\author{Orr Fischer%
	\thanks{Computer Science Department, Tel-Aviv University. Email: orrfischer@mail.tau.ac.il}
	\and
	Shay Gershtein%
	\thanks{Computer Science Department, Tel-Aviv University. Email: shayg1@mail.tau.ac.il}
	\and
	Rotem Oshman%
	\thanks{Computer Science Department, Tel-Aviv University. Email: roshman@mail.tau.ac.il}
}
\begin{document}
\newcommand{\CS}[1]{\fnstyle{C}\left(#1\right)}
\newcommand{\Cov}[1]{\rv{\fnstyle{Cov}}\left( #1 \right)}
\newcommand{\Top}{\fnstyle{Top}}
\newcommand{\Rep}{\fnstyle{Rep}}
\newcommand{\Vee}{\fnstyle{Vee}}
\newcommand{\Out}[1]{\fnstyle{Out}\left( #1 \right)}
\newcommand{\Good}{\fnstyle{Good}}
\newcommand{\embed}{\fnstyle{embed}}
\newcommand{\CC}{\fnstyle{CC}}
\newcommand\numberthis{\addtocounter{equation}{1}\tag{\theequation}}

\maketitle
% !TeX root = thesis.tex
In this paper we initiate the study of  property testing in simultaneous and non-simultaneous multi-party communication complexity, focusing on testing triangle-freeness in graphs.
We consider the \emph{coordinator} model, where we have $k$ players receiving private inputs, and a coordinator who receives no input; the coordinator can communicate with all the players, but the players cannot communicate with each other. In this model, we ask: if an input graph is divided between the players, with each player receiving some of the edges, how many bits do the players and the coordinator need to exchange to determine if the graph is triangle-free, or \emph{far} from triangle-free?

For general communication protocols, we show that $\tilde{O}(k(nd)^{1/4}+k^2)$ bits are sufficient to test triangle-freeness in graphs of size $n$ with average degree $d$ (the degree need not be known in advance). For \emph{simultaneous} protocols, where there is only one communication round, we give a protocol that uses $\tilde{O}(k \sqrt{n})$ bits when $d = O(\sqrt{n})$ and $\tilde{O}(k (nd)^{1/3})$ when $d = \Omega(\sqrt{n})$; here, again, the average degree $d$ does not need to be known in advance.
We show that for average degree $d = O(1)$, our simultaneous protocol is asymptotically optimal up to logarithmic factors.
For higher degrees, we are not able to give lower bounds on testing triangle-freeness, but we give evidence that the problem is hard by showing that finding an edge that participates in a triangle is hard, even when promised that at least a constant fraction of the edges must be removed in order to make the graph triangle-free.

\pagebreak
%\tableofcontents
\section{Introduction}
% !TeX root = main.tex

The field of property testing asks the following question:
%suppose we have a universe of objects, for example, graphs, represented using their adjacency matrix.
for a given property $P$, how hard is it to test whether an input \emph{satisfies} $P$, or is \emph{$\eps$-far from $P$}, in the sense that an $\eps$-fraction of its representation would need to be changed to obtain an object satisfying $P$?
Property-testing has received extensive attention, including graph properties such as connectivity and bipartiteness~\cite{GGR98}, properties of Boolean functions (monotonicity, linearity, etc.), properties of distributions, and many others \cite{DBLP:journals/fttcs/Ron09,Fischer01,Gol98}.
The usual model in which property-testing is studied is the \emph{query} model, in which the tester cannot ``see'' the entire input, and accesses it by asking local queries, that is by only viewing a single entry in the object representation at a time. The tester typically does not have at its disposal the possibility of making a "non-local" query whose answer depends on a substantial subset of the object's representation, which is a primary source of difficulty in property testing. For example, for graphs represented by their adjacency matrix, the tester might ask whether a given edge is in the graph, or what is the degree of some vertex.
The efficiency of a property tester is measured by the number of queries it needs to make.
One can also distinguish between \emph{oblivious} testers, which decide in advance on the set of queries, and \emph{adaptive} testers, which decide on the next query after observing the answers to the previous queries.
It is known that for many graph properties, one-sided oblivious testers are no more than quadratically more expensive than adaptive testers~\cite{GT03}.
%It is known that for a large class of graph properties in dense graphs, oblivious testers are quadratically-optimal assuming a one-sided error \cite{GT03}.

In this paper we study property testing from a different perspective, that of \emph{communication complexity}. We focus on property testing for graphs, and we assume that the input graph is divided between several players, who must communicate in order to determine whether it satisfies the property or is far from satisfying it.
Each player can operate on its own part of the input ``for free'', without needing to make queries; we charge only for the number of bits that the players exchange between them. This is on one hand easier than the query model, because players are not restricted to making local queries, and on the other hand harder, because the query model is centralized while here we are in a distributed setting.
This leads us to questions such as: does the fact that players are not restricted to local queries make the problem easier, or even trivial? Which useful ``building blocks'' from the world of property testing can be implemented efficiently by multi-party protocols? Does \emph{interaction} between the players help, or can we adopt the ``oblivious approach'' represented by \emph{simultaneous} communication protocols?

Beyond the intrinsic interest of these questions, our work is motivated by two recent lines of research. First,~\cite{CensorHillelFS16,FraigniaudRST16}, study property testing in the CONGEST model, 
and show that many graph property-testing problems can be solved efficiently in the distributed setting. As pointed out in~\cite{FraigniaudRST16}, existing techniques for proving lower bounds in the CONGEST model seem ill-suited to proving lower bounds for property testing. It seems that such lower bounds will require some advances on the communication complexity side, and in this paper we make initial steps in this direction.
Second, recent work has shown that many \emph{exact} problems are hard in the setting of multi-party communication complexity: Woodruff et al. \cite{Woodruff13} proved that for several natural graph properties, such as triangle-freeness, bipartiteness and connectivity, 
determining whether a graph satisfies the property essentially requires each player to send its entire input.
%in order to determine whether the input satisfies them, each player must essentially send a a message of the order of the number of vertices, or even the number of edges, which means the problem is maximally hard. \TODO{is this ok?}
We therefore ask whether weakening our requirements by turning to property testing can help.

In this work we focus mostly on the specific graph property of \emph{triangle-freeness}, an important property which has received a wealth of attention in the property testing literature.
It is known that in dense graphs (average degree $\Theta(n)$) there is an \emph{oblivious} tester for triangle-freeness which is asymptotically optimal in terms of the size of the graph (i.e., adaptivity does not help) \cite{AlonFKS00,fox2011new}, and~\cite{Alon:2006:TTG:1109557.1109589} also gives an oblivious tester for graphs with average degree $\Omega(\sqrt{n})$.
The closest parallel to oblivious testers in the world of communication complexity is \emph{simultaneous communication protocols}, where the players each send a single message to a referee,
and the referee then outputs the answer.
We devote special attention to the question of the simultaneous communication-complexity of testing triangle-freeness.

\subsection{Related Work}

Property testing is an important notion in many areas of theoretical computer science; see the surveys~\cite{DBLP:journals/fttcs/Ron09,Fischer01,Gol98} for more background.

%, and has been used in wide-ranging contexts, from probabilistically-checkable proofs to coding and cryptography.
%The first paper to study property testing in graphs is~\cite{GGR98}, and much work followed; we mention here the most directly related results, and refer to the surveys~\cite{DBLP:journals/fttcs/Ron09,Fischer01,Gol98} for more background.

Triangle-freeness, the problem we consider in this paper, is one of the most extensively studied properties in the world of property testing;
many different graph densities and restrictions have been investigated (e.g.,~\cite{AlonFKS00,Alon:2002:TSL:772614.772624,Alon:2003:TSD:780542.780644,Goldreich:1997:PTB:258533.258627}).
%In the dense model, \cite{AlonFKS00} showed that using the Regularity Lemma, triangle freeness (as well as many other properties) can be tested using a number of samples completely independent of the size $n$ of the graph, although in \cite{Alon:2002:TSL:772614.772624,Alon:2003:TSD:780542.780644} it is proven that the dependency on the distance parameter $\eps$ must be super-polynomial.
%In the bounded degree model,  \cite{Goldreich:1997:PTB:258533.258627} showed that testing triangle freeness can be solved with $O(\frac{d_{max}^2}{\eps})$ queries for graphs with maximal degree $d_{max}$.
%\paragraph{Testing triangle-freeness in the general model.}
Of particular relevance to us is triangle-freeness in the \emph{general model} of property testing, where 
the average degree of the graph is known in advance, but no other restrictions are imposed.
For this model,~\cite{Alon:2006:TTG:1109557.1109589} showed an upper bound of $\tilde{O}(\min\{\sqrt{nd}, n^{\frac{4}{3}}/d^{\frac{2}{3}}\})$ on testing triangle-freeness, and a lower bound of $\Omega(\max\{\sqrt{n}/d, \min\{d, n/d\}, \min\{\sqrt{d},n^{2/3}\}n^{-o(1)}\})$, both for graphs with a average degree $d$ ranging from $\Omega(1)$ up to $n^{1 - o(1)}$. For specific ranges of $d$, \cite{Ras06} and \cite{Gun06} improved these upper and lower bounds, respectively, by showing an upper bound of $O(\max\{(nd)^{4/9},n^{2/3}/d^{1/3}\}$ and a lower bound of $\Omega(\min\{(nd)^{1/3},n/d), n/d\})$.

Our simultaneous protocols use ideas, and in one case an entire tester, from~\cite{Alon:2006:TTG:1109557.1109589}, but implementing them in our model presents different challenges and opportunities. (Our unrestricted-round protocol does not bear much similarity to existing testers.) As for lower bounds, we cannot use the techniques from~\cite{Alon:2006:TTG:1109557.1109589} or other property-testing lower bounds, because they rely on the fact that the tester only has query access to the graph. For example, ~\cite{Alon:2006:TTG:1109557.1109589} uses the fact that a triangle-freeness tester with one-sided error must \emph{find} a triangle before it can announce that the graph is far from triangle-free (\cite{Alon:2006:TTG:1109557.1109589} also gives a reduction lifting their results to two-sided error).
In the communication complexity setting this is no longer true; there is no obvious reason why the players need to find a triangle in order to learn that the graph is not triangle-free.

\paragraph{Property testing in other contexts.}
Recently, the study of property testing has been explored in distributed computing \cite{Brakerski2011,CensorHillelFS16,FraigniaudRST16}. Among their other results, Censor-Hillel et al. \cite{CensorHillelFS16} showed that triangle-freeness can be tested in $O(\frac{1}{\eps^2})$ rounds in the CONGEST model; expanding this, $\cite{FraigniaudRST16}$ showed that testing $H$-freeness for any $4$ node graph $H$ can be done in $O(\frac{1}{\eps^2})$ rounds, and showed that their BFS and DFS approach fails for $K_5$ and $C_5$-freeness, respectively; ~\cite{FraigniaudRST16} does not give a general lower bound. There has also been work on property testing in the streaming model \cite{HP16}. The related problem of computing the exact or approximate number of triangles has also been studied in many contexts,
including distributed computing~\cite{Dolev12,Censor15,LeGall16,Drucker14}, sublinear-time algorithms (see~\cite{EdenLRS15} and the references therein), and streaming (e.g.,~\cite{Kallaugher17}).
Specifically,~\cite{Kallaugher17} gives a reduction which shows a lower bound on the space complexity of approximating the number of triangles in the streaming model; we apply their reduction here to show the hardness of testing triangle-freeness, by reducing from a different variant of the problem used in~\cite{Kallaugher17}.

\paragraph{Communication complexity.}
The multi-party number-in-hand model of communication complexity has received significant attention recently. In~\cite{Woodruff13} it is shown that several graph problems, including exact triangle-detection, are hard in this model. Many other exact and approximation problems have also been studied, including~\cite{Boutsidis16,Li14,Woodruff14,Braverman16,Woodruff12,Braverman13} and others.

Unfortunately, it seems that canonical lower bounds and techniques in communication complexity cannot be leveraged to obtain property-testing lower bounds; for discussion, see section~\ref{sec:discuss}.

\subsection{Our Contributions}

The contributions of this work are as follows:

\paragraph{Basic building-blocks.}
We show that many useful building-blocks from the property testing world can be implemented efficiently in the multi-player setting, allowing us to use existing property testers in our setting as well.
For some primitives --- e.g., sampling a random set of vertices --- this is immediate.
However, in some cases it is less obvious, especially when \emph{edge duplication} is allowed (so that several players can receive the same edge from the input graph).
We show that even with edge duplication the players can efficiently simulate a random walk, estimate the degree of a node, and implement other building blocks.

\paragraph{Upper bounds on testing triangle-freeness.}
For unrestricted communication protocols, we show that $\tilde{O}(k\sqrt[4]{nd}+k^2)$ bits are sufficient to test triangle-freeness, where $n$ is the size of the graph, $d$ is its average degree (which is not known in advance), and $k$ is the number of players.
When interaction is not allowed (simultaneous protocols), we give a protocol that uses $\tilde{O}( k\sqrt{n} )$ bits when $d = O(\sqrt{n})$,
and another protocol using $\tilde{O}( k\sqrt[3]{nd} )$ bits for the case $d = \Omega(\sqrt{n})$.
We also combine these protocols into a single \emph{degree-oblivious protocol}, which does not need to know the average degree in advance. (This is not as simple as might sound, since we are working with \emph{simultaneous} protocols, where we cannot first estimate the degree and then use the appropriate protocol for it.)

\paragraph{Lower bounds.}

Our lower bounds are mostly restricted to simultaneous protocols, although we first prove lower bounds for one-way protocols for two or three players, and then then ``lift'' the results to simultaneous protocols for $k \geq 3$ players using the symmetrization technique \cite{PVZ12}.

We show that for average degree $d = O(1)$, $\Omega(k \sqrt{n})$ bits are required to simultaneously test triangle-freeness, matching our upper bound. For higher degrees, we are not able to give a lower bound on testing triangle-freeness, but we give evidence that the problem is hard: we show that it is hard to \emph{find an edge that participates in a triangle}, even in graphs that are $\eps$-far from triangle free (for constant $\eps$), and where every edge participates in a triangle with (small) constant marginal probability.

\section{Preliminaries}
\label{section:preliminaries}
% !TeX root = main.tex
\label{sec:prelim}

\paragraph{Unrestricted Communication in the \textit{Number-in-Hand} Model}
The default model we consider is this work is the \textit{number-in-hand} model. In this model $k$ players receive private inputs $X_1, \ldots, X_k$ and communicate with each other in order to determine a joint function of their inputs $f(X_1 \ldots X_k)$. 
This is the most general model, as the number of rounds of communication is unrestricted. There are three common variants to this model, according to the mode of communication: 
the \textit{blackboard model}, where a message by any player is seen by everyone; the \textit{message-passing model}, where every two players have a private communication channel and each message has a specific recipient; the \textit{coordinator model}, which is the variant we consider in this paper, and define promptly.  

In the \emph{coordinator model} the players communicate over private channels with the coordinator, but they cannot communicate directly with each other. The protocol is divided into communication rounds. In each such round, the coordinator sends a message of arbitrary size to one of the players, who then responds back with a message. Eventually the coordinator outputs the answer  $f(X_1 \ldots X_k)$. For convenience, we assume that the players and the coordinator have access to \emph{shared randomness} instead of private randomness. Note that the players will make explicit use of the fact that the randomness is shared for common procedures like sampling, as the players can agree on which elements to sample simply by agreeing in advance (as part of the protocol) on how to interpret the public bits, and no interaction is required. (For protocols that use more than one round, it is possible to get rid of this assumption and use private randomness instead via Newman's Theorem~\cite{Newman91}, which costs at most additional $O(k \log{n})$ bits. For further details see~\cite{KushilevitzNisan, Newman91}).

The \emph{communication complexity} of a protocol $\Pi$, denoted $\CC(\Pi)$,
is the maximum over inputs of the expected number of bits exchanged between the players and the coordinator in the protocol's run.
For a problem $P$, we let $\CC_{k,\delta}(P)$ denote the best communication complexity of any protocol that solves $P$ with worst-case error probability $\delta$ on any input.

The coordinator model is roughly equivalent to the widely used message-passing model. More concretely, every protocol in the message-passing model can be simulated with a coordinator, incurring an overhead factor of at most $\log{k}$ by appending to each message the id of the recipient, to infrom the coordinator whom to forward this message to. The other direction can also be simulated efficiently, as in the message-passing model we can assign an arbitrary player to be the coordinator and run the protocol as it is. Although in this paper for convenience we consider the coordinator model, our results consequently apply for the message-passing model as well, up to a $\log{k}$ factor.

\paragraph{Simultaneous Communication}
Of particular interest to us in this work are \emph{simultaneous protocols}, which are, in a sense, the analog of oblivious property testers. This is the second primary model we investigate in addition to unrestricted communication.
In a simultaneous protocol, there is only one communication round, where each player, after seeing its input, sends a \emph{single message} to the coordinator (usually called the \emph{referee} in this context).
The coordinator then outputs the answer.
Any oblivious graph property tester which uses only edge queries (which test whether a given edge is in the graph or not) can be implemented by a simultaneous protocol,
but the converse is not necessarily true.

\paragraph{Communication complexity of property testing in graphs.}
we are given a graph $G = (V, E)$ on $n$ vertices, which is divided between the $k$ players, with each player $j$ receiving some subset $E_j \subseteq E$
of edges. More concretely, each player, $j$, receives the characteristic vector of $E_j$, where each entry corresponds to a single edge, such that if the bit is $1$ then that edge exists in $E$, and if the bit is $0$ it is unknown to the player whether it exists or not, as this entry might be $1$ in the input of a different player. The logical OR of all inputs results in the characteristic vector of the graph edges, $E$. Note that there is no guarantee for any vertex for a single player to have all its adjacent edges in its input, as is the case in models like CONGEST.
To make our results as broad as possible we follow the \emph{general model} of property testing in graphs (see, e.g.,~\cite{Alon:2006:TTG:1109557.1109589}):
we do not assume that the graph is regular or that there is an upper bound on the degree of individual nodes.
As in~\cite{Woodruff13}, edges may be duplicated, that is, the sets $E_1,\ldots,E_k$ are not necessarily disjoint.

The goal of a property tester for property $P$ is to distinguish the case where $G$ satisfies $P$ from the case where $G$ is $\eps$-far from satisfying $P$,
that is, at least $\eps |E|$ edges would need to be added or removed from $G$ to obtain a graph satisfying $P$.
An important parameter in our algorithms is the average degree, $d$, of the graph (also referred to as \textit{density});
for our upper bounds, we do not assume that $d$ is known, but our lower bounds can assume that it is known to the protocol up to a tight multiplicative factor of $(1 \pm o(1))$. Moreover, as in \cite{Alon:2006:TTG:1109557.1109589}, we focus on $d = \Omega(1)$ and $d \leq n^{1 - \nu(n)}$, where $\nu(n) = o(1)$, since for graphs of average degree $d = \Theta(n)$ there is a known solution whose complexity is independent on $n$ in the property-testing query model and consequently in our model as well. The case of $d = o(1)$, although not principally different, is ignored for simplicity, as its extreme sparsity makes it of less interest than any degree which is $\Omega(1)$.

\paragraph{Information theory.}
Our lower bounds use information theory to argue that using a small number of communication bits,
the players cannot convey much information about their inputs. 
For lack of space, we give here only the essential definitions and properties we need.

Let $(\rv{X}, \rv{Y}) \sim \mu$ be random variables. (In our lower bounds, for clarity, we adopt the convention that bold-face letters indicate random variables.)
To measure the information we learned about $\rv{X}$ after observing $\rv{Y}$,
we examine the difference between the \emph{prior} distribution of $\rv{X}$, denoted $\mu(\rv{X})$, and the \emph{posterior} distribution of $\rv{X}$ after seeing $\rv{Y} = y$, which we denote $\mu(\rv{X}|\rv{Y}=y)$.
We use KL divergence to quantify this difference:
\begin{definition}[KL Divergence]
	For distributions $\mu, \eta : \mathcal{X} \rightarrow [0,1]$,
	the \emph{KL divergence} between $\mu$ and $\eta$ is
	\begin{equation*}
		\Div{ \mu } {\eta} \coloneq \sum_{x \in \mathcal{X}} \mu(x) \log\left( \mu(x)/\eta(x) \right).
	\end{equation*}
\end{definition}
We require the following property, which follows from the superadditivity of information~\cite{CoverThomas}:
if $(\rv{X}_1, \ldots, \rv{X}_n, \rv{Y}) \sim \mu$ are such that $\rv{X}_1,\ldots, \rv{X}_n$ are independent, and $\rv{Y}$ can be represented using $m$ bits (that is, its \emph{entropy} is at most $m$), then
%\begin{equation*}
	$\E_{y \sim \mu(\rv{Y})} \left[ \sum_{i = 1}^n \Div{\mu(\rv{X}_i|\rv{Y}=y)}{\mu(\rv{X}_i)} \right] \leq m$.
%\end{equation*}
Here and in the sequel, $\mu(\rv{X}_i)$ denotes the marginal distribution of $\rv{X}_i$ according to $\mu$, and $\mu(\rv{X}_i|\rv{Y}=y)$ is the marginal distribution of $\rv{X}_i$ given $\rv{Y}=y$, and $\E_{y \sim \mu(\rv{Y})}$ denotes the expectation according to the distribution $\mu$.

\paragraph{Graph definitions and notation.}
We let $\deg(v)$ denote the degree of a vertex $v$ in the input graph, and for a player $j \in [k]$, we denote by $d^j(v)$ the degree of $v$ in player $j$'s input (the subgraph 
$(V, E_j)$).
\begin{definition}
	We say that a pair of edges $\set{ \set{u,v}, \set{v,w} } \subseteq E$ is a \emph{triangle-vee} if $\set{u,w} \in E$,
	and in this case we call $v$ the \emph{source} of the triangle-vee.
\end{definition}

\begin{definition}
	We say that an edge $e \in E$ is a \emph{triangle edge} if $G$ contains a triangle $T$, such that $e$ is an edge in $T$. 
\end{definition}

% \chapter{Model}
% \label{section:model}
% \input{03-model.tex}

% \chapter{Building Blocks}
% \label{section:blocks}
% \input{04-blocks.tex}

\section{Upper Bounds}
\label{section:upper}
% !TeX root = main.tex

All the solutions we present have a one-sided error, that is if a triangle is returned then it exists in $G$ with probability $1$. This holds even when the input is not $\eps$-far from being triangle free. Therefore, by solving the problem of triangle detection, we also solve triangle-freeness, as we never output a triangle in a triangle-free graph, and do output one with high probability whenever the gap guarantee holds. All algorithms have at most a small constant bound on the error, $\delta$. We also prove for some cases an improved complexity for several relaxations such as having the players communicate in the blackboard model, where each message is seen by all players, or the variant where the players are guaranteed there is no edge-duplication, such that each edge of the graph appears in exactly one input. Additionally, we assume that $k = O(poly(n))$, to simplify the complexity expressions.

\subsection{Building Blocks}

\label{sect:build_block}
We start by showing that the essential primitives used in the property testing setting (dense, sparse and general models combined) of graph problems are efficiently translatable into our communication complexity model, where the edges of the graph are scattered across $k$ inputs with possible multiplicity and as default the communication is unrestricted. This illustrates the added power packed in our communication complexity model, that can solve many problems with at most a logarithmic overhead factor, by simulating the PT solution, while for some problems, such as the one we will explore here, there is a significantly more efficient solution.
\begin{itemize}
\item \textbf{Querying a specific edge (check for its existence) }- This is one of the main primitives in the dense model. This can be done in our model in $O(k)$ by having each player send a bit to the coordinator indicating whether it is in its input, and the coordinator sends the answer bit to all players.

\item \textbf{Choosing uniformly a random edge adjacent to a given vertex $v$} - This is the main primitive in the sparse model. It can be simulated by utilizing the random bits to fix an random order, $P$, over all the $n-1$ potential edges adjacent to $v$, and have each player send the first edge in his input according to $P$ to the coordinator, who then sends everyone the first edge according to $P$ of all the edges he received. This costs $O(k\log(n))$. A random walk, which is a pivotal procedure in sparse property-testing, can be simulated by taking a random neighbor each step using this primitive. Note that the permutation was necessary so that edges with higher multiplicity would not be favored, as would happen in a naive implementation.

\item \textbf{Querying vertex degree} - this is an auxiliary query that is sometimes included in the general PT model. Without duplication this can be done trivially in $O(k\log(d(v))$, by having all players send the number of edges adjacent to $v$ in their input, and have the coordinator sum them up to get the result. With duplication, an exact answer costs $\Omega(kd(v))$ as it is at least as hard as solving disjointness, to ensure no over-counting. However an $\alpha$-approximation, for any $\alpha>1$ can be performed in efficiently as we promptly prove. We can also reduce the complexity in the no-duplication case by using an approximation, as we also show, and in many cases, such as triangle-detection, a constant approximation is good enough. 

\begin{theorem}\label{th:degree_approx}
	For any given vertex, $v$, the players can compute an $\alpha$-approximation, for a constant $\alpha > 1$ with probability at least $(1-\tau)$ and communication complexity of $O(k\log\log d(v) + k\log{k}\log{\log{k}}\log{\frac{1}{\tau}})$.
\end{theorem}
\begin{proof}
First each player, $P_i$, computes locally, $d_i(v)$, the number of edges adjacent to $v$ in his input, $E_i$, and sends $I_i$, the index of the MSB (the leftmost '1' bit in the binary representation of $d_i(v)$), to the coordinator, who then proceeds to compute the sum $d'=\sum\limits_{i \in [k]} 2^{I_i +1}$. This amounts to at most $2\sum\limits_{i \in [k]} 2^{I_i +1}$ and at least $\sum\limits_{i \in [k]} 2^{I_i +1}$, which is in itself a $K$-approximation of $d(v)$, where we can only over-count. Hence $\frac{d'}{2k} \leq d(v) \leq d'$. The coordinator announces $d'$ to each player and they proceed to the next step. The cost so far has been $O(k\log\log{d(v)})$.

In the second phase, the players start a $O(\log{k})$-round procedure, where in each round their decrease their guess, $d''$, of $d(v)$, by a multiplicative factor of $\sqrt{\alpha}$, having the starting guess be $d''= d'$. In each round they repeat independently an experiment of sampling possible edges (adjacent to $v$) and checking whether the sample contains an edge in $E$. They set a threshold for each round, and if the number of samples that contained an edge in $E$ exceeds the threshold, they stop and declare the value of $d''$ for that round as the approximated value. If the last guess is reached, the players output it without running the experiment. Note that the case of $d''=1$, therefore, is never checked, and we can assume that $d'' > 2$.

In each round, $r$, we denote $d''(r)$ as the value of the guess, $d''$, for that round, and $m(r)$ is the number of experiments they run. A single experiment is choosing into a set, $S(r)$, every neighbor of $v$ with probability $\frac{1}{d''(r)}$, using public randomness, and then each player sends the coordinator a bit indicating whether $S(r) \cap E_i = \emptyset$. Each round the players assume their guess is correct, and therefore compute $F(r)=(1 - (1-d''(r))^{d''(r)})$ as the probability of success in a single experiment, and thus the expected fraction of successes. The actual probability (as well as the expected fraction of successes) is $E(r)=(1 - (1-d''(r))^{d(v)})$.

We wish to prove that, with high probability, if $d''(r) > \alpha d(v)$, then the number of successes does not exceed the threshold, which is $\frac{F(r)}{c}$, where $c$ is a small constant who value we will determine later. In that case we have $E(r) < F(r)$ and $\frac{F(r)}{E(r)} \geq \frac{(1 - (1-d''(r))^{d''(r)})}{(1 - (1-d''(r))^{\frac{d''(r)}{\alpha}})}$, which is lower bounded by, $(1 +\beta_1(\alpha))$, where $\beta_1$ is a constant dependent only on $\alpha$. The players may assume the lowest possible value for $d''$, in order for the expression to be dependent only on $\alpha$, as for $d''=2$ we get a small constant bigger than $1$, and as $d''$ tends to infinity it increases to $\frac{1-\frac{1}{e}}{1-(\frac{1}{e})^{\frac{1}{\alpha}}}$. Therefore, by choosing c to be small enough (less than the square root of the difference) and running a number of experiments dependent on $\beta_1$ and $\tau$, a Chernoff bound yields a $1-\tau$ bound on the probability of exceeding the threshold. We wish to reduce the error by a $\Theta(\log{k})$ factor, to $\frac{\tau}{2\log k}$, so that we may use the union bound to ensure this deviation doesn't happen in any round where $d''(r) > \alpha d(v)$, and by a Chernoff argument an increase to $m(r)$ by a factor of $O(\log\log k)$ suffices.  

On the other hand, we show that the first guess where $d''(r) < \frac{d(v)}{\sqrt{\alpha}}$, we exceed the threshold with probability at least $(1-\tau/2)$. Combined with what we proved in the previous case, this proves that with probability at least $(1- \tau)$ we stop at a guess which is within the bounds of an $alpha$-approximation of the $d(v)$, as required. 

Now we have $d''(r) < \frac{d(v)}{\sqrt{\alpha}}$ that implies $E(r) > F(r)$ and $\frac{F(r)}{E(r)} \geq \frac{(1 - (1-d''(r))^{d''(r)})}{(1 - (1-d''(r))^{\sqrt{\alpha}d''(r)})}$, which is upper bounded by constant, $\beta_2(\alpha)$, dependent only on $\alpha$. Therefore, $m(r) = \Theta(\log\log k)$ is more than enough for a constant bound on a constant deviation small enough not to reduce the number of successes below the threshold.

The total complexity of each round, therefore, is $O(k \log\log{k})$, and summing across all rounds we get $O(k \log{k} \log\log{k})$. The overall complexity of the algorithm is $O(k\log\log d(v) + k\log{k}\log{\log{k}})$. 

\end{proof}

\begin{lemma}
	In the no-duplication variant, for any given vertex, $v$, the players can compute a $\alpha$-approximation of $d(v)$, for $alpha = O(1)$, with complexity $O(k\log\log{\frac{d(v)}{k}})$.
\end{lemma}
\begin{proof}
	Each player,$P_i$, computes locally, $d_i(v)$, the number of edges adjacent to $v$ in his input, $E_i$, and sends to the coordinator the $(\log{\frac{1}{\alpha}})$ most significant bits along with the index of the cutoff, which takes $O(\log\log{d_i(v)})$ bits to represent. The coordinator assumes all the missing bits are zeros, and which makes this an $\alpha$-approximation of $d_i(v)$ for each player, when we can only under-count, and thus the sum of all this approximation, is also an $\alpha$-approximation. Note that since there is no duplication, the worst case, by convexity, is when all players have a $\frac{1}{k}$-fraction of the edges, which implies the bound stated in the lemma. 
\end{proof}

Note that this approximation procedure can be applied to any subset of vertex pairs, including estimating the total number of edges in the graph, and not only to the specific set of all possible edges adjacent to a given vertex. More generally, this solves the problem of estimating the number of distinct elements in a set. 

\item \textbf{Choosing uniformly a random edge} - this  usually can't be performed efficiently in the PT model (unless the standard model is augmented), and is commonly replaced by choosing a random vertex, and then a random neighbor from the adjacency list. In our model however, we can once again use randomness to fix an order over the edges and each player sends to the coordinator its highest ranked edge, which is $O(\log(n))$, and across all players this sums up to $O(k \log{n})$. The coordinator chooses the highest ranked edge, and posts it to all players.

\item \textbf{Selecting all edges of a subgraph induced by $V' \subset V$ }- this also cannot be performed efficiently in the PT model, as it requires querying all possible vertex pairs, or going over all relevant adjacency lists. In our model, on the other hand, it is possible by having each player post all edges (in the blackboard model the players post in turns so as to not repeat the same edge) of the relevant subgraph in his input. Let $m$ denote the number of edges in the relevant subgraph, then the complexity is $O(km \log{n})$, and it is the same with simultaneous communication (except that only the referee will know the answer). The complexity is reduced to $O(m \log{n})$ if there is no edge-duplication, or if the players communicate using a blackboard. If $m$ is significantly smaller than $|V'|^2$, then our procedure is more efficient.
This is particularly relevant when implementing a BFS. It can be done in $O(n\log{n})$ by having all players post all the neighbors of the currently examined vertex.
% Consider for example the problem of testing graph bipartiteness. The optimal solution in the PT model is to select randomly $\Theta(1/\epsilon)$ vertices, run a BFS on them, and have each player check if he has an edge that creates a cycles in the BFS tree (if not, we accept). We can implement this in $O(\frac{1}{\epsilon}\log\frac{1}{\epsilon})$, whereas the PT lower bound is $\Omega(\epsilon^{-\frac{3}{2}})$. 

\end{itemize}

\subsection{Input analysis}
Prior to discussing our proposed algorithms, we analyze the properties of the input - a graph $\eps$-far from being triangle free. Our pivotal tool for this analysis is bucketing. We partition $V$ into buckets, such that for $1\leq i \leq \lfloor\log_3{n} +1 \rfloor$ we have $B_i = \{v\in V |\quad 3^{i-1} \leq deg(v) < 3^{i}  \}$, whereas $B_0$ is the bucket of singletons. Note that there are less than $\log{n}$ buckets. Let $d^-(B_i) = 3^{i-1}$ and $d^+(B_i) = 3^{i}$ denote respectively the minimal and maximal bounds on degrees of vertices in $B_i$. We use $d(B_i)$ to mean any degree in that range, when the $3$ factor is negligible, and refer to it as the degree of the bucket. We say an edge is adjacent to a bucket, if it is adjacent to at least one of its vertices. Additionally, we call a set of triangle-vees disjoint if any two of them are either edge disjoint or originate from a different vertex. Note that for simplicity we ignore some rounding issues, and avoid using floor or ceiling values. 

We are interested in buckets that contain many vertices that participate in a large number of triangles. Towards that end, we introduce the following definition, and analyze its properties. 

\begin{definition}[full bucket] We call $B_i$ a \textit{full bucket} if the edges adjacent to it contain a set of $\frac{\epsilon n d}{2\log{n}}$ disjoint triangle-vees. Let $B_{min}$ denote the bucket with the lowest degree of all the full buckets.
\end{definition}

\begin{observation}\label{ob:full_bucket_exists}
By the pigeonhole principle there is at least one full bucket, as there are at least $\eps n d$ disjoint triangle-vees.
\end{observation}

\begin{lemma}[size of a full bucket]\label{lm:bucket_size}
If $B_i$ is a full bucket then: 
$$\frac{\eps nd}{\log{n}\cdot d^+(B_i)} \leq |B_i| \leq \min\{n, \frac{2nd}{d^-(B_i)}\}$$
when the upper bound holds regardless of $B_i$ being full. 
\end{lemma}
\begin{proof}
The number of disjoint triangle-vees in a full bucket is at least $\frac{\eps nd}{2\log{n}}$. Therefore, the lower bound pertains to the extreme case when all vertices have the maximal degree $d^+(B_i)$, that consists entirely of $d^+(B_i)/2$ disjoint triangle-vees, thus reaching the sum $\frac{\eps nd}{\log{n}}$ with as few vertices as possible. The upper bound follows from the opposite extreme when each vertex contributes as little as possible which is $d^-(B_i)$, and there are at most $\frac{2nd}{d^-(B_i)}$ such vertices as it would amount to $nd$ edges, the total number of edges in the graph (the $2$ factor follows from counting each edge twice). 
\end{proof}

\begin{definition}[full vertex]\label{df:full_vertex}
We call a vertex, $v$, a \textit{full vertex} if at least an $\frac{\epsilon}{12\log{n}}$-fraction of the edges adjacent to it are a set of disjoint triangle-vees. Additionally, let 
\begin{equation*}
		F(B_i) = \set{ v \in B_i \st \text{v is full}},
	\end{equation*}
be the set of full vertices in $B_i$. And let 
\begin{equation*}
		F(V) = \set{ v \in V \st \text{v is full}},
	\end{equation*}
be the set of all full vertices in $V$.
\end{definition}

Full vertices vertices play a vital role in finding a triangles, as they, by definition, participate in many disjoint triangles. We, therefore, prove several useful lemmas about their incidence, as we are interested in identifying such vertices, preferably using sampling.  

\begin{lemma}\label{lm:fraction_full_vertex_full_bucket}
At least an $\frac{\epsilon}{12\log{n}}$-fraction of the vertices in a full bucket, $B_i$, are full. 
\end{lemma}
\begin{proof}
We prove that otherwise there are less than $ \frac{\eps nd}{2\log{n}}$ triangle-vees adjacent to it, which contradicts it being full.
This holds even if we disregard any double-counting, and assume the bucket has the maximal size of $\frac{2nd}{d^-(B_i)}$, and all its vertex degrees are $d^+(B_i)$ (both assumptions can not hold simultaneously, but this only strengthens our proof). 

The total contribution to the count of triangle-vees coming from non-full vertices is less than:
$$\frac{1}{2} \cdot \frac{2nd}{d^-(B_i)} \cdot d^+(B_i) \cdot \frac{\epsilon}{12\log{n}} = 
\frac{1}{4\log{n}}\eps nd$$. 

Each full vertex can contribute at most $\frac{d^+(B_i)}{2}$ vees to the count, and if we assume the fraction of full vertices is less than $\frac{\epsilon}{12\log{n}}$, it amounts to less than:

$$\frac{\epsilon}{12\log{n}} \cdot \frac{2nd}{d^-(B_i)} \cdot \frac{d^+(B_i)}{2} = 
\frac{1}{4\log{n}}\eps nd$$. 

overall all vertices combined contribute less than the required $\frac{\eps nd}{2\log{n}}$ to the disjoint triangle-vees count. 
\end{proof}

Lemmas \ref{lm:bucket_size} and \ref{lm:fraction_full_vertex_full_bucket} imply the following corollary:

\begin{corollary}\label{cr:num_of_full_vertices}
The number of full vertices in a full bucket, $B_i$, is at least:
$$|F(B_i)| \geq \frac{\eps^2 \cdot d}{12\cdot \log^2{n} \cdot d^+(B_i)} \cdot n$$
\end{corollary}

Next, we single out a set of buckets in proximity to a given bucket, that will play a special role in our algorithm. 

\begin{definition}[r-neighbourhood of a bucket]\label{df:neighbourhood}
Let $r \in \mathbb{N}$ such that $r \leq \log_3{n}$. We call 
$$N_r(B_i) = \set{B_j \st j \geq (i- \log_3{r})}$$
the \textbf{r-neighborhood} of bucket $B_i$, that is the set of all buckets of higher degrees, itself, and the $\log_3{r}$ buckets right below it in the degree ranking. Additionally, we call 
$$N(B_i) = B_{i-1} \cup B_i \cup B_{i+1}$$
the \textbf{neighborhood} of bucket $B_i$. 
\end{definition}

\begin{lemma}\label{lm:ratio_full_vertices_neighbourhood}
Let $B_i$ be a full bucket. We prove the following lower bound on the ratio of the number of full vertices in it to the combined size of the buckets in its neighborhood: 
\begin{equation}
\frac{|F(B_i)|}{|N(B_i)|} \geq \frac{\eps^2}{312 \cdot \log^2{n}}
\end{equation}
\end{lemma}
\begin{proof}
For any $j$ we have the upper bound $\frac{2nd}{d^-(B_j)}$ on the size of $B_j$ as we have proven in lemma \ref{lm:bucket_size}. Therefore, we have the following upper bound on the sum of bucket sizes: 

\begin{align*}
&
|N(B_i)| = |B_{i-1}| + |B_{i}| + |B_{i+1}|
\\
&\leq 
3 \cdot \frac{2nd}{d^-(B_i)} + \frac{2nd}{d^-(B_i)} + \frac{1}{3} \cdot \frac{2nd}{d^-(B_j)} = 
\\
&= \frac{26 \cdot n d}{3 \cdot d^-(B_i)} = \frac{26 \cdot n d}{3^i}
\end{align*}

Additionally, we have $\frac{\eps^2 \cdot d n}{12\cdot \log^2{n} \cdot d^+(B_i)} $ as a lower bound on $|F(B_i)|$ (corollary \ref{cr:num_of_full_vertices}). Hence, we get the following lower bound on the ratio: 
\begin{equation}
\frac{|F(B_i)|}{|N(B_i)|} \geq 
\frac{\frac{\eps^2 \cdot d n}{12\cdot \log^2{n} \cdot 3^i}}{\frac{26 \cdot n d}{3^i}} 
= \frac{\eps^2}{312 \cdot \log^2{n}}
\end{equation}
\end{proof}

\begin{lemma}\label{lm:ratio_neighbourhood_r_neighbourhood}
Let $B_i$ be a full bucket. We prove the following lower bound on the ratio of the number of full vertices in it to the combined size of the buckets in its r-neighborhood: 
\begin{equation}
\frac{|F(B_i)|}{\sum\limits_{B_j \in N_r(B_i)}|B_j|} \geq \frac{\eps^2}{108 \cdot \log^2{n} \cdot r}
\end{equation}
\end{lemma}
\begin{proof}
For any $j$ we have the upper bound $\frac{2nd}{d^-(B_j)}$ on the size of $B_j$ as we have proven in lemma \ref{lm:bucket_size}. Therefore, we have the following upper bound on the sum of bucket sizes: 

\begin{align*}
&
\sum\limits_{B_j \in N_r(B_i)}|B_j| = 
\sum\limits_{j=(i-\log_3{r})}^{\lfloor\log_3{n} +1 \rfloor}|B_j| \leq 
\sum\limits_{j=(i-\log_3{r})}^{\lfloor\log_3{n} +1 \rfloor} \frac{2nd}{3^{j-1}}
\\
& < \sum\limits_{j \geq (i-\log_3{r})} \frac{2 n d}{3^{j-1}} \leq 
\frac{6 n d \cdot r}{3^i} \cdot \sum\limits_{m \geq 0} \frac{1}{3^m} \leq 
\frac{6 n d \cdot r}{3^i} \cdot \frac{3}{2}
\\
&= \frac{9\cdot n d r}{3^i}
\end{align*}

Additionally, we have $\frac{\eps^2 \cdot d n}{12\cdot \log^2{n} \cdot 3^i} $ as a lower bound on $|F(B_i)|$ (corollary \ref{cr:num_of_full_vertices}). Hence, we get the following lower bound on the ratio: 
\begin{equation}
\frac{|F(B_i)|}{\sum\limits_{B_j \in N_r(B_i)}|B_j|} \geq 
\frac{\frac{\eps^2 \cdot d n}{12\cdot \log^2{n} \cdot 3^i }}{\frac{9\cdot n d r}{3^i}} 
= \frac{\eps^2}{108 \cdot \log^2{n} \cdot r}
\end{equation}
\end{proof}

Now we show that we can efficiently sample edges adjacent to a full vertex in order to detect a triangle-vee.

\begin{lemma}[Extended Birthday Paradox]\label{lm:extended_birthday_paradox}
Let v be a vertex of degree $d(v) \geq 2$, such that at least $\alpha d(v)$, for $\alpha \geq 2/d$, of the edges adjacent to it are a set of disjoint triangle-vees. It is enough to sample each edge independently with probability $p = c \cdot \frac{1}{\sqrt{\alpha\cdot d(v)}}$, where $c = 4 \cdot \sqrt{\ln{\frac{1}{\delta'}}}$, in order for the sampled set to contain a triangle-vee with probability at least $(1-\delta')$. 
\end{lemma}
\begin{proof}
The probability of any specific triangle-vee to be sampled is $p^2$. By the linearity of expectation the expected number of triangle-vees sampled is $p^2 \cdot \frac{\alpha d(v)}{2} = \frac{c^2}{2}$. By a Chernoff bound the probability that less than one triangle-vee has been sampled is less than $e^{-\frac{c^2}{4} \cdot (1-\frac{2}{c^2})^2} \leq e^{-\frac{c^2}{16}} = \delta'$.
\end{proof}

\begin{corollary}\label{cr:full_birthday_paradox}
Let $v$ be a full vertex of degree $d(v)$. By sampling independently every edge adjacent to it with probability $p = 4 \cdot \sqrt{\ln{\frac{6}{\delta}}} \cdot \sqrt{\frac{12 \log{n}}{\eps \cdot d(v)}}$, we find a triangle-vee with probability at least $(1- \delta/6)$.
\end{corollary}
\begin{proof}
We get this trivially by plugging $\alpha = \frac{\epsilon}{12\log{n}}$ and $\delta' = \delta/6$ in lemma \ref{lm:extended_birthday_paradox}. 
\end{proof}

Finally, we prove that there are many triangle-vees adjacent to vertices of degree $O(\sqrt{nd})$, such that we can focus only on such vertices, and adjust our analysis accordingly.

\begin{definition}\label{df:d_h}
Let $V_h$ denote the subset of $V$ that contains all the vertices with degree at least $d_h = \sqrt{nd /\eps}$. Let $E_h \subset E$ denote all edges with both endpoints in $V_h$. Finally, let $V_l=V\backslash V_h$, and let $G_l$ denote the resulting graph when $E_h$ is removed from $G$. 
\end{definition}

\begin{lemma}\label{lm:num_vees_in_v_l}
$G_l$ is $\frac{\eps}{2}$-far from being triangle-free, and there are at least $\epsilon n d/2$ disjoint triangle-vees adjacent to vertices in $V_l$.
\end{lemma}
\begin{proof}
Because $|V_h|\leq \frac{n d}{d_h}= \sqrt{\epsilon nd}$, it follows that $E_h < \frac{\epsilon nd}{2}$, hence if all edges in $E_h$ are removed, at least $\frac{\epsilon nd}{2}$ additional edges need to be removed from $G$ for it to be triangle-free, as at least $\eps nd$ are required in total by definition. This also implies that $\frac{\epsilon nd}{2}$ of the triangle-vees are adjacent to a vertices in $V_l$. 
\end{proof}

\begin{definition}
Let $d_l = \frac{\eps d}{2\log{n}}$. 
\end{definition}

\begin{lemma}\label{lm:bounds_on_degree_of_B_min}
We have the following bound on degree of vertices in $B_{min}$:
$$
d_l \leq d^-(B_{min}) \leq d_h
$$
\end{lemma}
\begin{proof}
The lower bound follows a simple counting argument, as even if all $n$ vertices are in $B_{min}$, if $d^-(B_{min}) < d_l$, then the total number of triangle-vees adjacent to $B_{min}$ is less than $\frac{\eps d}{2\log{n}} n$ which contradicts it being full. 
The upper bound is implied by lemma \ref{lm:num_vees_in_v_l} and the pigeonhole principle. 
\end{proof}

\begin{corollary}\label{cr:charesteristics_B_min}
We have:
	\begin{enumerate}%[(1)]
	\item $|B_{min}| \geq = \frac{\epsilon^{\frac{3}{2}}}{3 \cdot \log{n}} \cdot \sqrt{nd}$, 
	\item $|F(B_{min})| \geq \frac{\eps^{\frac{5}{2}}}{12\cdot \log^2{n} \cdot 3} \cdot \sqrt{nd}$.
	\end{enumerate}
\end{corollary}

\begin{proof}
If plug the upper bound from lemma \ref{lm:bounds_on_degree_of_B_min} into lemma \ref{lm:bucket_size} and corollary \ref{cr:num_of_full_vertices} we get the first and second clause of this corollary respectively. 
\end{proof}

\subsection{Unrestricted Communication}
\label{sect:unbound_app}

The first protocol we present requires \emph{interaction} between the players, and exploits the following advantage we have over the query model: suppose that the players have managed to find a set $S \subseteq E$ of edges that contains a ``triangle-vee'' --- a pair of edges $\set{u, v}, \set{v, w} \in S$ such that $\set{u,w} \in E$ (but $\set{u,w}$ is not necessarily in $S$). Then even if $S$ is very large, the players can easily conclude that the graph contains a triangle: each player examines its own input and checks if it has an edge that closes a triangle together with some vee in $S$, and in the next round informs the other players. Thus, in our model, finding a triangle boils down to finding a triangle-vee. (In contrast, in the query model we would need to query $\set{u,w}$ for every 2-path $\set{u,v}, \set{v,w} \in S$, and this could be expensive if $S$ is large.)

% \begin{observation}
% If by a certain point in the protocol two different edges of the same triangle have been posted to all players, the algorithm is completed successfully in one additional step, since the player holding the third triangle-edge posts it immediately. This means we can focus on solving the problem of posting a triangle vee (the two edges may be posted in different steps and the players may not know of the posted vee being a triangle-vee), as long as we make sure all players know of all posted edges. 
% \end{observation}

Our goal is to find a full vertex. Once that is obtained, we proved in lemma \ref{lm:extended_birthday_paradox} we can efficiently sample a relatively small subset of its edges to find a triangle-vee, thus successfully ending the algorithm. More concretely, If $v$ is a full vertex, then sampling each of its edges with probability $p_{d(v)} = \Theta(\sqrt{\log n/d(v)})$ will reveal a triangle-vee with constant probability. 

Note that $\deg(v)$ may be significantly higher than the average degree $d$ in the graph, so we cannot necessarily afford to sample each of $v$'s edges with probability $p_{d(v)}$; we need to find a \emph{low-degree} vertex which is full. Towards that end, we proved in lemma \ref{lm:num_vees_in_v_l} we can focus only on vertices of degree at most $d_h = O(\sqrt{nd})$. 

With that in mind, we proceed to present our strategy for finding a full vertex. We first describe the core of the algorithm in a relatively detail-free manner, to emphasize the intuitive narrative leading us throughout the procedure. This will be followed by a rigorous analysis of the full algorithm.

How can the players find a full vertex? A uniformly random vertex is not always likely to be full --- there might be a small dense subgraph of relatively high-degree nodes which contains all the triangles. In order to target such dense subgraphs, we use \emph{bucketing}: we partition the vertices into buckets, with each bucket $B_i$ containing the vertices with degrees in the range $[3^i, 3^{i+1})$. We want to find a full bucket, and sample its vertices, as we proved that many of them are full. Of course, we cannot know in advance which bucket is full; we must try all the buckets. Hence, we iterate over the buckets in an increasing order of their associated vertex degree, up until $d_h$, and for each bucket, assume it is full, and then sampling enough of its vertices, relying on the lower bound we proved for the number of full vertices in a full bucket, for the sample to include a full vertex. Then, we sample the edges of each vertex, which for the full vertex will result in discovering a triangle-vee with high probability. Although, our assumption can be wrong in many cases, we have proven that there is at least one full bucket, $B_i$, in that range of degrees, hence the assumption will be correct at least once, which is enough for our algorithm to succeed with high probability. 

It remains to describe, given $B_i$ is a full bucket, how we can sample a random vertex from it, that is,
a random vertex with degree in the range $[3^i, 3^{i+1})$. We cannot do that, precisely, but we can come close. Because the edges are divided between the players, no single player initially knows the degree of any given vertex. However, by the pigeonhole principle, for each vertex $v$ there is some player that has at least $\deg(v)/k$ of $v$'s edges, and of course no player has more than $\deg(v)$ edges for $v$.

Let $\tilde{B}_i^j \coloneq \set{ v \in V \st 3^i / k \leq d^j(v) \leq 3^{i+1} }$
	be the set of vertices that player $j$ can ``reasonably suspect'' belong to bucket $i$, where $d^j$ denotes the  degree of vertex $v$ in the input of player $j$,
	and let $\tilde{B}_i \coloneq \bigcup_j \tilde{B}_j$.
	By the argument above, $B_i \subseteq \tilde{B}_i$. Also, $\tilde{B}_i \subseteq N_k(B_i)$,
	since the total degree of any vertex selected cannot be smaller than $3^i/k$.
	Therefore, sampling uniformly from $\tilde{B}_i$ is a good proxy for sampling from $B_i$, although we may 
	also hit adjacent buckets. Nevertheless, we have proven that a full bucket must be large, and hence constitutes at least roughly a $\frac{1}{k}$-fraction of $\tilde{B}_i$ and $N_k(B_i)$. Hence a uniformly random sample will yield a vertex from $B_i$ with probability at least roughly $1/k$, and $\tilde{\Theta}(k)$ samples yield a vertex from $B_i$ with high probability.
    
\begin{lemma}\label{lm:samples}
	Let $B_i$ be a full bucket. It suffices to sample uniformly with replacement $m =  \ln{(\frac{6}{\delta})}\cdot \frac{108 \cdot \log^2{n} \cdot r}{\eps^2}$ vertices from $N_r(B_i)$, in order for the sampled set to contain a vertex from $F(B_i)$ with probability at least $(1- \frac{\delta}{6})$.
\end{lemma}
\begin{proof}
	The probability of sampling a vertex from $F(B_i)$ is $p=\frac{\eps^2}{108 \cdot \log^2{n} \cdot r}$ according to Lemma \ref{lm:ratio_neighbourhood_r_neighbourhood}. Therefore, the probability of not having a vertex from $F(B_i)$ after $m$ samples is bounded by $(1-p)^m \leq e^{-m p} = \delta/6$. 
\end{proof}

To implement the sampling procedure we need two components:
first, we need to be able to sample uniformly from $\tilde{B}_i$. The difficulty here is that each vertex $v \in \tilde{B}_i$ can be known to a different number of players --- possibly only one player $j$ has $v \in \tilde{B}_i^j$, possibly all players do. If we try a naive approach, such as having each player $j$ post a random sample from $\tilde{B}_i^j$, then our sample will be biased in favor of vertices that belong to $\tilde{B}_i^j$ for many players $j$.
	Our solution is to impose a random order on the nodes in $\tilde{B}_i$ by publicly sampling a permutation $\pi$ on $V$ (this is done by interpreting the random bits as a permutation on $V$),
	and we then choose the \emph{smallest} node in $\tilde{B}_i$ with respect to $\pi$.
	This yields a uniformly random sample, unbiased by the number of players that know of a given node.
	We call this procedure $\mathrm{SampleUniformFrom\tilde{B}_i}$.
    
\begin{algorithm}
	\caption{\small SampleUniformFrom$\tilde{B}_i$}
	\label{alg:sampleu}
	\begin{algorithmic}[1]
		\State $\pi \gets $ random permutation on $V$
		\State Each player $j$ sends the coordinator the first vertex in $\tilde{B}_i^j$ with respect to $\pi$
		\State The coordinator outputs the first vertex with respect to $\pi$ of all the vertices it received 
	\end{algorithmic}
\end{algorithm}    

Our sample is too large to treat every sampled vertex as if it is a full vertex from $B_i$. Sampling edges for each vertex is too costly and wasteful, since it is possible that only a $\frac{1}{k}$-fraction of the sampled vertices are even in $B_i$. The second component, therefore, verifies that a sampled node indeed belongs to $B_i$. We cannot do that exactly, but we can come sufficiently close. We compute a $\sqrt{3}$-approximation of the degree of the sampled node, as explained in Theorem \ref{th:degree_approx}, and discard vertices whose approximate degree does not match $N(B_i)$. This substantially reduces the size of the sampled set without discarding any vertex from $B_i$. We call this procedure $\mathrm{ApproxDegree}(v)$.

	The protocol for player $j$ is sketched in Algorithm~\ref{alg:unbounded}.
	Here $N = \tilde{\Theta}(k)$ is the number of samples from $\tilde{B}_i$ required to produce a sample from $B_i$ with good probability.
	Following the procedure described in Algorithm~\ref{alg:unbounded},
	the coordinator sends all the edges he received to all the players, and the players then check their own inputs
	for an edge that closes a triangle with some triangle-vee sent by the coordinator.
	With high probability, a triangle-vee \emph{is} discovered, and the protocol completes in the next round.

\begin{algorithm}
	\begin{algorithmic}
	\State For each $i = 0,\ldots,\log n$:
	\State $\ell \leftarrow 0$
	\State Repeat until $\ell \geq N$:
	\State \qquad $v \leftarrow \mathrm{SampleUniformFrom\tilde{B}_i}$
	\State \qquad $\bar{d}(v) \leftarrow \mathrm{ApproxDegree}(v)$
	\State \qquad If $d^-(B_i) / \sqrt{3} \leq \bar{d}(v) \leq \sqrt{3} d^+(B_i)$:
	\State \qquad \qquad $\ell \leftarrow \ell + 1$
	\State \qquad \qquad Jointly generate a public random set $S \subseteq V$, where each $u \in S$ with iid probability $p_{\bar{d}(v)}$
	\State \qquad \qquad %If $|E_j \cap \left( \set{v} \times S \right)| \leq s_{\bar{d}(v)}$, 
			Send $E_j \cap \left( \set{v} \times S \right)$ to the coordinator
\end{algorithmic}
\caption{Code for player $j$}
\label{alg:unbounded}
\end{algorithm}

We move on to a more rigorous analysis of the sampling parameters and the complexity. First we compute the number of vertices we need to uniformly sample from $N(B_i)$ to find a full vertex, so that we can bound the number of vertices we examine while retaining a high probability of preserving a full vertex. 

\begin{lemma}\label{lm:samples_full_vertex}
	Let $B_i$ be a full bucket. It suffices to sample uniformly with replacement $m = \ln{(\frac{6}{\delta})}\cdot \frac{312 \cdot \log^2{n}}{ \eps^2}$ vertices from $N(B_i)$, in order for the sampled set to contain a vertex in $F(B_i)$ with probability at least $(1- \frac{\delta}{6})$.
\end{lemma}
\begin{proof}
	The probability of sampling a vertex from $F(B_i)$ is $p=\frac{312 \cdot \log^2{n}}{\delta \cdot \eps^2}$ according to Lemma \ref{lm:ratio_full_vertices_neighbourhood}. Therefore, the probability of not having a vertex from $F(B_i)$ after $m$ samples is bounded by $(1-p)^m \leq e^{-m p} = \delta/6$. 
\end{proof}

We now present the procedure $\mathrm{GetFullCandidates}(B_i)$. Let $q = \ln{(\frac{6}{\delta})}\cdot \frac{108 \cdot \log^2{n} \cdot k}{\eps^2}$. we use $q$ as a bound on the total number of samples, and we also bound the number of samples that pass the degree approximation criteria. Both bounds are needed to ensure worst case complexity.

\begin{algorithm}
	\label{GetFullCandidates}\small
	\caption{\small GetFullCandidates($B_i$)}
	\begin{algorithmic}[1] 
		\State $count \gets $ 0
		\State $C \gets  \emptyset$
		\State Do until  $count = q$ or $|C| = \ln{(\frac{6}{\delta})}\cdot \frac{312 \cdot \log^2{n}}{ \eps^2}$
        \State \quad $count \gets $ $count + 1$
		\State \quad $v \gets $ SampleUniformFrom$\tilde{B}_i$
		\State \quad compute $d'(v)$, a $\sqrt{3}$-approximation of $d(v)$, such that the error probability is at most $\frac{\delta}{3q}$
		\State \quad if $\frac{d^-(B_i)}{\sqrt{3}} \leq d'(v) \leq \sqrt{3}d^+(B_i)$ then add $v$ to $C$
		\State output C
	\end{algorithmic}
\end{algorithm}

\begin{lemma}
	The complexity of GetFullCandidates($B_i$) is $O(k^2\cdot\log^4{n}\log\log{n})$, and in the no-duplication variant it is $O(k^2 \cdot \log^3{n})$.
\end{lemma}
\begin{proof}
	Each vertex the players send to the coordinator costs $O(\log{n})$, therefore the overall complexity of choosing uniformly a vertex from $U(B_i)$ is $O(k\log{n})$. According to Theorem \ref{th:degree_approx}, the complexity of a constant approximation of $d(v)$ for vertex $v$, when the error bound is $\Theta(\frac{1}{q})$, is $O(k\cdot\log{k}\cdot\log\log{k}\cdot(\log\log{n}+\log{k})) = O(k \cdot \log^2{n}\cdot \log\log{n})$, and in the no-duplication variant it is $O(k\log\log{\frac{n}{k}})$. The number of iterations is at most $q = O(k \log^2{n})$, hence the total complexity is $O(k^2\cdot\log^4{n}\log\log{n})$, and in the no-duplication variant it is $O(k^2 \cdot \log^3{n})$.
\end{proof}

\begin{lemma}\label{lm:find_full}
	If $B_i$ is a full, then $C$ contains a vertex from $F(B_i)$ along with a correct $\sqrt{3}$-approximation of its degree with probability at least $1-\frac{2\delta}{3}$. 
\end{lemma}
\begin{proof}
	First, note that a union bound implies that all $O(q)$ vertex approximations were correct with probability at least $(1-\frac{\delta}{3})$.
	Due to the symmetric process of sampling in algorithm \ref{alg:sampleu}, the vertices are chosen uniformly from $\tilde{B}_i$. Lemma \ref{lm:samples}, when substituting $r=k$, implies that $q$ uniform samples is enough to sample a vertex from $F(B_i)$ with probability at least $(1-\frac{\delta}{6})$. And since $\tilde{B}_i \subseteq N_k(B_i)$, sampling from $\tilde{B}_i$ can only improve our probability of finding a full vertex. 
	Additionally, lemma $\ref{lm:samples_full_vertex}$ assures us that $\ln{(\frac{6}{\delta})}\cdot \frac{312 \cdot \log^2{n}}{ \eps^2}$ uniform samples from $N(B_i)$ is enough to encounter a full vertex with probability at least $(1-\frac{\delta}{6})$. Note that if all degree approximations are in the guaranteed range, then all vertices sampled from $B_i$ are added to $C$, and all other vertices that are added are in $N(B_i)$. Since our sampling process is symmetric and therefore uniform, vertices sampled into $C$ have at least the same probability of containing a full vertex, as in the case of sampling uniformly from all of $N(B_i)$, and thus $\ln{(\frac{6}{\delta})}\cdot \frac{312 \cdot \log^2{n}}{ \eps^2}$ samples suffice. By a union bound argument, the probability of all approximations being correct, and both sampling processes sufficing for finding a full vertex, is at least $(1- \frac{\delta}{3}) - 2\cdot \frac{\delta}{6}) = (1-\frac{2}{3}\delta)$. Finally, if both sampling processes would encounter a full vertex, we find it no matter which stopping condition made us halt, therefore it is also the probability of containing a full vertex in the output.
\end{proof}

By this point we know how to efficiently obtain, given $B_i$ is a full bucket, a small sample of vertices that is likely to contain one full vertex from $B_i$. All that is left is to iterate over the sampled set, sampling edges adjacent to each vertex, such that for the full vertex the sample will contain a triangle-vee. 

\begin{algorithm}
	\caption{\small SampleEdges($v$)}
	\begin{algorithmic}[1]
		\label{SampleEdges}\small
		\State $S \gets $ sample every possible edge adjacent to $v$ with probability $p = 4 \cdot \sqrt{\ln{\frac{6}{\delta}}} \cdot \sqrt{\frac{12 \log{n}}{\eps \cdot \frac{d'(v)}{3}}}$
		\State each player $j$ sends the coordinator $S \cap E_j$ if this set is of size at most $(1+\frac{18}{d'(v)\cdot p}\ln{\frac{6}{\delta}})\cdot \sqrt{3}d'(v)p$
		\State the coordinator outputs $S \cap E$
	\end{algorithmic}
\end{algorithm}

\begin{algorithm}
	\caption{\small FindTrinagleVee($B_i$)}
	\begin{algorithmic}[1]
		\label{FindTrinagleVee}\small
		\State $C \gets $ GetFullCandidates($B_i$)
		\State for each $v \in C$ let $S \gets $ SampleEdges($v$) and then the coordinator posts all the edges to all the players.
	\end{algorithmic}
\end{algorithm}

\begin{lemma}
	The communication complexity of FindTrinagleVee($B_i$) is $O(k\cdot \log^{\frac{3}{2}}{n} \cdot \sqrt{d(B_i)}+k^2\cdot\log^4{n}\log\log{n})$
\end{lemma}
\begin{proof}
	The cost of GetFullCandidates is $O(k^2\cdot\log^4{n}\log\log{n})$.
	Each edge required $O(\log{n})$ bits to identify, and since there is a limit on the size of the set the players can send, the overall complexity is $O(k\cdot \log^{\frac{3}{2}}{n} \cdot \sqrt{d(B_i)}+k^2\cdot\log^4{n}\log\log{n})$.
\end{proof}

\begin{lemma}\label{lm:full_bucket_correct}
	If $B_i$ is a full bucket then the players find a triangle with probability at least $1-\delta$ using procedure $\mathrm{FindTrinagleVee(B_i)}$.
\end{lemma}
\begin{proof}
	According to lemma \ref{lm:find_full}, $C$ contains a full vertex with probability at least $(1-\frac{2}{3}\delta)$, and the approximation procedure of its degree gave a correct output. And according to corollary \ref{cr:full_birthday_paradox}, if $v$ is a full vertex, then sampling each of its edges with probability $p$ suffices to sample a triangle vee, with probability at least $(1-\frac{\delta}{6})$. Moreover, the expected size of the sampled set is $d(v)\cdot p$, and by a Chernoff bound the probability of sampling more the cutoff size specified in step $2$ of the sampling algorithm is at most $\delta/6$. Overall, with probability at least $(1-\delta)$, one of the vertices sampled will be full, with the players knowing approximately its degree, and hence sampling enough of its edges so that they find a triangle-vee, and do not need to send a set above the cutoff size. When that happens, one of the players will respond to the coordinator with the third edge completing the triangle-vee into a triangle, and the algorithm will end successfully. 
\end{proof}

All that is left is to find a full bucket. This is achieved by iterating all the relevant buckets. Note the common theme which is that the players do not know how successful each sampling process was up until the algorithm terminates. They do not know which bucket is full, which of the sampled vertices are full, or which of its adjacent edges belong to triangles. Since they examine all bucket, all sampled nodes (after preliminary filtering), and send all sampled edges, that knowledge is redundant, as when they encounter a full bucket and a full vertex they will be treated as such as a working assumption. Our analysis culminates in algorithm  $\mathrm{FindTrinagle}(G)$, which nests inside it all the procedure we presented so far.

\begin{algorithm}
	\caption{\small FindTrinagle($G$)}
	\begin{algorithmic}[1]
		\label{FindTrinagle}\small
		\State Run FindTrinagleVee($B_i$) for every bucket starting the first bucket such that $d^-(B_i) \geq d_l$ and until the last bucket where $d^+(B_i) \leq d_h$
	\end{algorithmic}
\end{algorithm}

\begin{theorem}
	If the input graph is $\eps$-far from being triangle free, then the players can find a triangle  with probability at least $(1 - \delta)$ and complexity $O(k\sqrt[4]{nd}\log^{5/2}n+k^2\log^5n\log\log{n})= \tilde{O}(k\sqrt[4]{nd}+k^2)$. The complexity is in fact $\tilde{O}(k\sqrt{d(B_{min})}+k^2)$ w.p. at least $1 - \delta$.  
\end{theorem}
\begin{proof}
According to lemma \ref{lm:bounds_on_degree_of_B_min}, one of the iterations of FindTriangleVee is performed on $B_{min}$, which is a full bucket. According to lemma \ref{lm:full_bucket_correct}, the players find a triangle in that iteration with probability $(1-\delta)$. The maximal complexity of an iteration grows as $d(B_i)$ grows. The number of iterations is $O(\log{n})$. Therefore, with probability at least $(1 - \delta)$, the overall complexity is $O(k\sqrt{d(B_{min})}\log^{5/2}n+k^2\log^5n\log\log{n})$. However, if an error occurs, then the players will go over all buckets in the given range, thus in the worst case the complexity is $O(k\sqrt{d_h}\log^{5/2}n+k^2\log^5n\log\log{n})= \tilde{O}(k\sqrt[4]{nd}+k^2)$.
\end{proof}

\begin{corollary}
	There is a one-sided error protocol for testing triangle-freeness with communication complexity $O(k\sqrt[4]{nd}\log^{5/2}n+k^2\log^4n\log\log{n})= \tilde{O}(k\sqrt[4]{nd}+k^2)$. The complexity is in fact $\tilde{O}(k\sqrt{d(B_{min})}+k^2)$ w.p. at least $1 - \delta$.
\end{corollary}

\begin{corollary}
The variant where the players are not given $d$ has the same complexity.
\end{corollary}
\begin{proof}
The players can compute $2$-approximations of $d_h$ and $d_l$, denoted by $d'_h$ and $d'_l$, respectively, and then use $d'_l/2$ and $2d'_h$ as the boundaries for the iteration condition. The added complexity is negligible, and the error can be reduced to an arbitrarily small constant. The rest of the algorithm remains the same, and does not rely on any knowledge of $d$.
\end{proof}

\begin{theorem}
		In the blackboard model, if the input graph is $\eps$-far from being triangle free, then the players can find a triangle  with probability at least $(1 - \delta)$ and complexity $O(\sqrt[4]{nd}\log^{5/2}n+k^2\log^5n\log\log{n})= \tilde{O}(\sqrt[4]{nd}+k^2)$. The complexity is in fact $\tilde{O}(\sqrt{d(B_{min})}+k^2)$ w.p. at least $1 - \delta$. The same protocol also solves testing of triangle-freeness.
\end{theorem}
\begin{proof}
	In the blackboard model, posting the edges of the sub-procedure $\mathrm{SampleEdges}(v)$ can be implemented more efficiently, having each player post the edges on the blackboard in turns, ensuring no edge is posted twice. This saves a $k$ factor in the communication cost with regards to the coordinator model.  
\end{proof}

\subsection{Simultaneous Communication}
In the simultaneous model, the players cannot interact with each other --- they send only one message to the referee,
and the referee then outputs the answer.
This rules out our previous approach, as exposing a triangle-vee does not help us if the players cannot then check their inputs for an edge that completes the triangle.
Indeed, the simultaneous model is closer to the query model in spirit. Accordingly, we include triangle-freeness testers of~\cite{Alon:2006:TTG:1109557.1109589}, but show that we can implement them more efficiently in our model. Moreover, we achieve roughly the same complexity without knowing the average degree in advance. 

We present separate algorithms for the case of $d=\Omega(\sqrt{n})$ and $d=O(\sqrt{n})$, referred to as high and low degrees, respectively. when $d=\Theta(\sqrt{n})$, both algorithms are essentially identical. We conclude with an algorithm that works for the more general case where $d$ is unknown to the players.

\subsubsection{high degrees}
For graphs with average degree $\Omega(\sqrt{n})$, the tester from~\cite{Alon:2006:TTG:1109557.1109589} samples a uniformly random set $S \subseteq V$ of $\Theta(\sqrt[3]{n^2/d})$ vertices, queries all edges in $S^2$, and checks if the subgraph exposed contains a triangle. It is shown in~\cite{Alon:2006:TTG:1109557.1109589} that if the graph is $\eps$-far from triangle-free, then the subgraph induced by $S$ will contain $\Theta(1)$ triangles in expectation, and the variance is small enough to ensure small error.

We can implement this tester easily, and in our model it is less expensive: instead of querying all pairs in $S^2$, the players simply send all the edges from $S^2$ in their input, paying only for edges that exist and not for edges that do not exist in the graph. The set $S$ is large enough that the number of edges in the subgraph does not deviate significantly from its expected value, $\Theta\left( (nd)^{1/3} \right)$.

We present an algorithm of complexity $O({k(nd)^{1/3}})$. We later show, in the lower-bounds section, that for average degree $\Theta(\sqrt{n})$ this is tight for 3 players.

 \begin{algorithm}
 	\caption{\small FindTringleSimHigh($G$)} \label{alg:sim_high}\small
 	\begin{algorithmic}[1]
		\State $S \gets$ a uniformly random set of vertices of size $c \sqrt[3]{\frac{n^2}{\eps d}}$, for a sufficiently large $c$
		\State players send all edges in the subgraph induced by the vertices in $S$. If the number of edges to be sent by a player exceeds $l = \frac{|S|^2}{n^2}\frac{4}{\delta}nd$, send any $l$ edges.
		\State The Referee checks whether the union on edges it received contains a triangle, and outputs accordingly.
 	\end{algorithmic}
 \end{algorithm}
 
\begin{theorem}
	The problem of triangle detection, when $d = \Omega(\sqrt{n})$ is known to the players, can be solved with communication cost of $O(k(nd)^{1/3}\log{n})$ and a constant error. 
\end{theorem}
\begin{proof}
Let $\delta$ denote the required bound on the error, and let $V_S$ denote the subgraph of $G$ induced by $S$. In \cite{Alon:2006:TTG:1109557.1109589} it was shown that for a sufficiently large $c$, the probability of $V_S$ not containing a triangle is arbitrarily small, and for our purposes taken as $\delta/2$. Additionally, the probability of each edge appearing in $V_S$ is $\frac{|S|\cdot|S-1|}{n\cdot(n-1)}$, thus by the linearity of expectation the expected number of edges in $V_S$ is $\frac{|S|\cdot|S-1|}{n\cdot(n-1)} nd < \frac{l}{\frac{2}{\delta}}$. Finally, by a Markov argument we get that the probability of the number of edges in $V_S$ exceeding $l$ is at most $\delta/2$. When that happens, even if every player has all the edges, none of them exceeds $l$, and overall, by a union bound argument, we get that the referee receives all edges in $V_S$ and they contain a triangle, with probability at least $(1-\delta)$. The complexity is at most $O(kl\log{n}) = O(k(nd)^{\frac{1}{3}}\log{n})$.
\end{proof}

\begin{corollary}
The problem of triangle-detection in the no-duplication variant can be solved in the simultaneous model with a constant error, $\delta$, such that the complexity is $O((nd)^{1/3}\log{n})$ with probability at least $(1 - \delta)$, and the worst case complexity is $O(k(nd)^{1/3}\log{n})$.
\end{corollary}
\begin{proof}
The players would use algorithm \ref{alg:sim_high} as in the general case, thus the worst case complexity is the same. But as we proved, with probability at least $1-\delta$ the total number of edges in the subgraph is $O((nd)^{1/3})$, and since there is no duplication, that is also the total number of edges sent. 
\end{proof}

\subsubsection{low degrees}

For density $d = o(\sqrt{n})$, the approach above no longer works, as the variance is too large. To illustrate this, consider a graph with $d$ vertices of degree $\Theta(n)$, which are the sources of $\Theta(n d)$ triangle-vees, such that all triangles have at least one such node. If we were to sample vertices uniformly at random, we need to sample $\Theta(n/d)$ vertices in order for the subgraph to contain a triangle.
However, whereas in the query model we would need to make $\Theta(n^2/d^2)$ queries to learn the entire subgraph induced by the set we sampled, in our model we proceed as follows (using ideas from~\cite{Alon:2006:TTG:1109557.1109589}, which require adaptivity there, and deploying them in a different way): 
let $S$ be the set of $\Theta(n/d)$ uniformly-random vertices. We sample another smaller set, $R$, of $\Theta(\sqrt{n})$ vertices, and we send all edges in $R \times (S \cup R)$. If indeed there is a small set of high-degree vertices participating in most of the triangles, then with good probability we will have one of them in $S$, and by the birthday paradox, one of its triangles will have its other two vertices in $R$. On the other hand, if the triangles are spread out ``evenly'', then the subgraph $R \times R$ will probably contain one. The expected size of $R \times (S \cup R)$ is $O(\sqrt{n})$, and we show that w.h.p. the total communication is $\tilde{O}(k \sqrt{n})$.

Note that both our solutions work for $d = \Theta(\sqrt{n})$, and for this density they are essentially the same:
both sets, $S$ and $R$, are of size $\Theta(n/d) = \Theta(d) = \Theta\left( (nd)^{1/3} \right)$,
so for $d = \Theta(\sqrt{n})$ the second protocol is not very different from the first.
We can also show that if edge duplication is not allowed, a factor of $k$ is saved in the communication complexity with high probability.

\begin{algorithm}
	\caption{\small FindTriangleSimLow($G$)}
	\label{alg:sim_low}
	\begin{algorithmic}[1]
		\State  $S \gets$ sample each vertex with probability $p_1 = \min\{\frac{c}{d}, 1\}$
		\State  $R \gets$ sample each vertex with probability $p_2 = \frac{c}{\sqrt{n}}$
		\State Players send all to the referee edges with one endpoint in $R$ and the second endpoint $R \cup S$. If the number of such edges in the input of a player exceeds $q = 2c^2(\sqrt{n} + d) \cdot \frac{2}{\delta}$, that player sends any $q$ edges. 
		\State The Referee checks whether the union of the edges it received contains a triangle, and outputs accordingly.
	\end{algorithmic}
\end{algorithm}
Where $c$ is a constant to be determined later.

\begin{theorem}
The problem of triangle detection in the simultaneous model when $d = O(\sqrt{n})$ is known to the players, can be solved with communication cost of $O(k\sqrt{n}\log{n})$ and with constant error.
\end{theorem}
\begin{proof}
	We show that algorithm \ref{alg:sim_low} is such a solution. Let $\delta$ denote the required constant bound on the error, let $V_{R}$ denote the graph induced by $R$, and let $V_{RS}$ denote the graph on $R \cup S$ that includes all edges with at least one endpoint in $R$. 
	Since each player sends at most $q$ edges, the complexity of the algorithm is $O(qk\log{n}) = O(k\sqrt{n}\log{n})$. 
    
The probability of a given edge appearing in $V_{R}$ is $p_2^2 = \frac{c^2}{n}$, therefore by linearity of expectation, the expected number of edges in $V_{R}$ is at most $nd \cdot \frac{c^2}{n} = c^2d$. The probability of any edge having one endpoint in $S$ and the other in $R$ is at most $2p_1p_2$, therefore by linearity of expectation, the expected number of such edges is at most $2ndp_1p_2\leq2c^2\sqrt{n}$. Overall, we get that the expected number of edges in $V_{RS}$, which are the only edges the players may send, is at most $\frac{q}{(\frac{2}{\delta})}$, and by a Markov argument we get that with probability at least $(1-\frac{\delta}{2})$ the number of edges in $V_S$ does not exceed $q$, thus all players can send all the edges they have in $V_{RS}$. 

We show that with probability at least $(1-\frac{\delta}{2})$ the edges in $V_{RS}$ contain a triangle, which via a union bound proves that the referee will receive a triangle with probability at least $(1 - \delta)$. 
	Recall $G_l$, the graph defined in definition \ref{df:d_h} to be the subgraph on all edges adjacent to at least one vertex of degree at most $d_h$. For the purpose of analysis, ignore all triangle that are not in $G_l$. According to lemma \ref{lm:num_vees_in_v_l}, $G_l$ is $\frac{\eps}{2}$-far from being triangle-free, and thus it contains a family of $\frac{\eps}{6}nd$ edge-disjoint triangles. Fix such a family, $T$. Note that in any triangle in $G_l$, at most one vertex in can be of degree higher than $d_h$. Further restricting the number of triangles counted, we only count triangles where at least two vertices of degree at most $d_h$ were sampled into $R$, which implies that if the third vertex is of degree higher than $d_h$ it must be sampled into $S$. Let $X$ be a random variable equal to the number of triangles in $V_{RS}$ with the aforementioned restrictions. 
	The probability of such a triangle to be sampled is at least $p_1p_2^2$, therefore the 
	\begin{equation}
	E[X] \geq \frac{\eps}{6}ndp_1p_2^2 = \frac{\eps}{6}c^2
	\end{equation}. 
    
    We now bound the variance of $X$. For each $t \in T$, let $X_t$ be an indicator of $t$ being sampled (with our restrictions). Let $d_T(v) \leq d(v)$ denote the degree of vertex $v$ when only edges of $T$ are left in the graph. Observe that if two triangles have no vertices in common, then they are selected independently. Additionally, note that two triangles in $T$ can have at most one vertex in common, as all triangles are edge disjoint. The probability of two triangles with a joint vertex being both sampled can be split into two cases, one where the joint vertex is in $S$, and the second case is when it is sampled into $R$. 
    
  The probability when the joint vertex is sampled into $S$, is $p_1p_2^4$. The number of triangles that can have vertex $v$ in common and sampled into $S$ is at most $\binom{d_T(v)/2}{2}$. Since $\sum\limits_{v \in V} d_T(v) \leq 2nd$, by convexity we get $\sum\limits_{v \in V} \binom{d_T(v)/2}{2} \leq \sum\limits_{i = 1}^{2d} \binom{n/2}{2} \leq 2d \cdot \frac{n^2}{8}$. 
  
	As for the case when the common vertex, $v$, is sampled into $R$ (note that it means that $d_T(v) < d_h$), the probability of both triangles being sampled is $p_1^2p_2^3$. The number of vertices of degree $d_h$ is at most $\frac{2nd}{d_h}$, and once again by convexity we get that the number of such triangles is smaller than $\sum\limits_{v \in V_h} \binom{d_T(v)/2}{2} \leq \sum\limits_{v \in V_h} \binom{d_h/2}{2} \leq \frac{2nd}{d_h} \cdot \frac{d_h^2}{8}$. 
    
 Therefore, the variance of $X$ is, for $d>c$, bounded by: 
	\begin{align*}
	&
	Var[X] \leq \sum\limits_{v \in V} \binom{d_T(v)/2} {2} p_1p_2^4 +  \sum\limits_{v \in {V \cap R}} \binom{d_T(v)/2}{2} p_1^2p_2^3 
	\\
	& \leq 
	2d \cdot \frac{n^2}{8}p_1p_2^4 + \frac{2nd}{d_h} \cdot \frac{d_h^2}{8} p_1^2p_2^3
    \\
    & \leq 
    2d \cdot \frac{n^2}{8}\frac{c}{d}(\frac{c}{\sqrt{n}})^4 + \frac{2nd\eps}{\sqrt{dn}} \cdot \frac{nd}{8\eps} (\frac{c}{d})^2 (\frac{c}{\sqrt{n}})^3 
    \\
   & = 
   \frac{c^5}{4} + \frac{c^5}{4\sqrt{d}} \leq \frac{c^5}{2}
	\end{align*}
	
	and similarly for $d \leq c$ (which implies that $p_1 = 1$, $S = V$ and $d = Theta(1)$) we get that the variance is also bounded by $\frac{c^5}{2}$.
	
	We conclude by employing a Chebyshev bound:
	\begin{align*}
	&
	\Pr(X < 1) \leq \Pr(|X -\E[X]| \geq \frac{1}{2}\E[X])  \leq 
	\frac{Var^2(X)}{((1/2)\E[X])^2} \leq \frac{8\eps c^4}{18c^5}  = \frac{4}{9c}< \frac{\delta}{2}
	\end{align*}
	Where the last inequality follows by taking $c = \frac{8}{9\delta}$.
Therefore, $V_{RS}$ contains at least one triangle with probability at least $(1-\delta/2)$.
\end{proof}
\begin{corollary}
	The complexity of the no-duplication variant is $O(\sqrt{n}\log{n})$ with probability at least $(1-\delta)$, and the worst case complexity is $O(k\sqrt{n}\log{n})$.
\end{corollary}
\begin{proof}
	The players would use algorithm \ref{alg:sim_low} as in the general case, thus the worst case complexity is the same. But as we proved, with probability at least $1-\delta$ the total number of edges in the subgraph between an endpoint of $R$ and $S$ is $O(\sqrt{n})$, and the total number of edges in the subgraph of $R$ is at most $O(d)$ and since there is no duplication, that is also the total number of edges sent, which implies complexity $O(\sqrt{n})$. 
\end{proof}

\subsubsection{Degree Oblivious Algorithm}
We start with a high-level overview of how we combine the protocols above and modify them so that they can be used
without advance knowledge of the degree.
The challenge here is that no single player can get a good estimate of the degree from their input, and since the protocol is simultaneous,
the players must decide what to do without consulting each other.
The natural approach is to use $\log n$ exponentially-increasing ``guesses'' for the density, covering the range $[1,n]$, and try them all; however, if we do this we will incur a high cost for guesses that imply examining a larger sample than needed.
We, therefore, take a more fine-grained approach.

Our first observation is that some players \emph{can} make a reasonable estimate of the global density, although they do not \emph{know} that they can. Let $\bar{d}^j$ denote the average degree in player $j$'s input $E_j$, and let us say that player $j$ is \emph{relevant} if $\bar{d}^j \geq (\eps/(4k))d$, and \emph{irrelevant} otherwise. If we eliminate all the irrelevant players and their inputs, the graph still remains $(\eps/2)$-far from triangle-free, so we can afford to ignore the irrelevant players in our analysis --- except for making sure that their messages are not too large.

Since players cannot know if they are relevant, all players assume that they are. Based on the degree $\bar{d}^j$ that player $j$ observes, it knows that \emph{if} it is relevant, then the average degree in the graph is in the range $D_j = [\bar{d}^j, \Theta(k \bar{d}^j)]$.
We fix in advance an exponential scale $\set{ 2^i }_{i=0}^{\log n}$ of guesses for the density,
and execute in parallel $\log n$ instances of triangle-freeness protocols, one for each degree $2^i$.
However, each player $j$ only participates in the $O(\log{k})$ instances corresponding to density guesses that fall in $D_j$,
and sends nothing for the other instances. The protocols are the two algorithms we have presented for a known average degree, with some modification.
For relevant players, we know that the true density falls in their range $D_j$, so they will participate in the ``correct'' instance. For irrelevant players, we do not care, and their message size is also not an issue: their density estimate is too low, and the communication complexity of each instance increases with the density it corresponds to.

If we are not careful, we may still incur a blow-up in communication, as relevant players may use guesses lower than the true density by a factor of $k$, which increases the size of the sample beyond what is necessary. However, by carefully assigning each player $j$ a communication budget depending on $\bar{d}^j$ , we can eliminate the blow-up, and match the degree-aware protocol up to polylogarithmic factors.

We now move on to a detailed analysis of the algorithm, which relies on an integration of modified versions of the algorithms we presented in the non-oblivious sections. First, we alter the algorithm for high-degrees, such that instead of sampling $|S|= \sqrt[3]{\frac{n^2}{\eps d}} = \Theta(\frac{n^{2/3}}{d^{1/3}})$ vertices without replacement, we sample each vertex independently (with replacement) with probability $\Theta(\frac{|S|}{n}) = \Theta(nd^{-1/3})$. Additionally, we remove the cap on the number of edges allowed to be sent from both algorithms (high and low degrees). Let $Alg_{high}$ denote the modified algorithm for high degrees, and $Alg_{low}$ - for low degrees. We provides figures for both.

 \begin{algorithm}
 	\caption{\small $Alg_{High}(G)$} \label{alg:alg_high}\small
 	\begin{algorithmic}[1]
		\State $S \gets$ a uniformly random set of vertices of size $c \sqrt[3]{\frac{n^2}{\eps d}}$, for a sufficiently large $c$
		\State players send all edges in the subgraph induced by the vertices in $S$. 
		\State The Referee checks whether the union on edges it received contains a triangle, and outputs accordingly.
 	\end{algorithmic}
 \end{algorithm}
 
\begin{algorithm}
	\caption{\small $Alg_{Low}(G)$}
	\label{alg:alg_low}
	\begin{algorithmic}[1]
		\State  $S \gets$ sample each vertex with probability $p_1 = \min\{\frac{c}{d}, 1\}$
		\State  $R \gets$ sample each vertex with probability $p_2 = \frac{c}{\sqrt{n}}$
		\State Players send all to the referee edges with one endpoint in $R$ and the second endpoint $R \cup S$. 
		\State The Referee checks whether the union of the edges it received contains a triangle, and outputs accordingly.
	\end{algorithmic}
\end{algorithm}

\begin{lemma}
$Alg_{high}$ detects a triangle with a small constant error.	
\end{lemma}
\begin{proof}
By choosing sufficiently large constants, we ensure that the expected number of triangles is not lower than the expected value before the alteration. The bound on the deviation, follows the same analysis as in the proof of the original algorithm, and analogous to the correctness proof of $Alg_{low}$, as when bounding the variance in the number of edge-disjoint triangles, only triangles with one common vertex are dependent. In terms of complexity, we once again use a Markov argument to claim that the total number of edges in the sampled subgraph does not exceed the expectation by a large constant factor with high probability. 
\end{proof}

$Alg_{low}$ obviously also remains correct, as removing the cap could only increase the chances of detecting a triangle. Moreover, both algorithms have the same complexity as before with probability at least $(1-\delta)$, as implied by their respective complexity analysis in the previous section. We will reinstall modified caps further into our analysis.

The main sub-procedure the players utilize in both algorithms is choosing jointly, via public randomness, a subset $S \subseteq V$, such that each vertex is chosen independently with probability $p$, and then posting all edges in their inputs with both endpoints in $S$. We show that the number of edges each player has in $S$ does not significantly exceed the expectation, given his average degree.

\begin{lemma}\label{lm:deviation_bound}
Let $S \subseteq V$ denote a set where each vertex was sampled with probability $p$. The number of edges player $j$ has in the subgraph $V_S$ induced by $S$ is $O(n\bar{d^j}\cdot p^2 \cdot log{n}\log(k\log{n}))$, such that this holds for all players with high probability $\Theta(1)$.
\end{lemma}
\begin{proof}
Recall the bucketing partition we used in the section of input analysis. For a given player $j$, we partition $V$ into $O(\log{n})$ buckets as we did before, only this time according to $d^j(v)$ and not $d(v)$. Since each vertex is chosen independently, we may utilize the Chernoff bound to claim that when sampling with probability $p$, the degree of each sampled vertex is reduced from $d^j(v)$ to $O(p\cdot d^j(v)\log(nk)) = O(p\cdot d^j(v)\log{n})$, with probability of error at most $O(\frac{1}{nk})$. Therefore, by the union bound this holds true for all $n$ vertices in the inputs of all $k$ players, with a small constant error. 

Next, we also claim that the number of vertices chosen from each bucket, $B$, is at most $O(p \cdot |B|\log(k\log{n}))$, with probability of error at most $O(\frac{1}{k\log{n}})$, once again due to a Chernoff argument. A union bound implies this holds true for all $O(\log{n})$ buckets for all $k$ players with a small constant error. Overall, if the number of vertices chosen from each bucket deviates by at most a $O(\log(k\log{n}))$ factor from the expectation, and each vertex degree is reduced to a size that deviates by at most an $O(\log{n})$ factor from its expectation, we get that for all $k$ players the number of edges they have in the sampled subgraph deviates from its expected value (given $\bar{d^j}$), by at most a $O(\log{n}\log(k\log{n}))$ factor.  
\end{proof}

For a given guess of the average degree $d'$, let $p(d')$ denote the probability with which the players need to sample each vertex. Recall that all relevant players will include $p(d)$ in their range of guesses (the $2$ factor difference is asymptotically insignificant), hence the union of all the sampled edges includes a triangles with high probability. We note that for our purposes, an increase of the degree guess, $d'$, by a factor of $2$ decreases the sampling probability, $p(d')$, also by at most a factor of 2, therefore there is no dangerous super-constant blowup in the sampling probability that would otherwise incur an asymptotic overhead on the complexity.

Observe that the guess for the average degree varies inversely as the corresponding sampling probability and thus expected sample size.

We first discuss the case where for player $j$, $\bar{d}^j = \Omega(\sqrt{n})$, which implies $d = \Omega(\sqrt{n})$. 
The player performs simultaneously $O(\log{k})$ algorithms each with a different guess, $d'$, of the average degree. For high degrees this means $p(d') = \Theta(\sqrt[3]{nd'})$.

We now prove that the complexity bound on the message for each player can remain roughly the same, with the error remaining constant. More concretely, each player limits separately each of the $O(\log{k})$
simultaneous algorithms by sending at most $O(\sqrt[3]{n\bar{d^j}}\log{n}\log{(k\log{n})})$ edges. This implies that the complexity bound of this player over all its simultaneous instances is $O(\sqrt[3]{nd}\log^2{n}\log{k}\log{(k\log{n})})$. 

\begin{lemma}
A complexity bound of sending at most $O(\sqrt[3]{n\bar{d^j}}\log{n}\log{(k\log{n})})$ edges for each instance of $Alg_{high}$ suffices for the instance, pertaining to the correct guess, to send all the edges in its corresponding subgraph.
\end{lemma}
\begin{proof}
Let $r(j) = \frac{d}{\bar{d^j}}$ denote the ratio between the correct average degree and the player's observed average degree. The expected number of edges player $j$ has in the sampled subgraph for a correct guess is 
$$\Theta(\frac{n \bar{d^j}}{(nd)^{2/3}}) = \Theta(\frac{{(n\bar{d^j})}^{1/3}}{r(j)^{2/3}}) \leq O(\sqrt[3]{n\bar{d^j}}) \leq O(\sqrt[3]{nd})$$
where we used the fact that $r(j) = \Omega(1)$. Therefore, player $j$ can limit the edge budget of each algorithm with a bound of $O(\sqrt[3]{n\bar{d^j}}\log{n}\log{(k\log{n})})$ following lemma \ref{lm:deviation_bound}, with all $k$ players not exceeding the bound with high probability. 
\end{proof}

This lemma along with the fact that we've shown that the sample pertaining to the correct guess contains a triangle with high probability implies correctness with constant error.

Now we deal with the case $\bar{d^i} \leq \sqrt{n}$. As in the previous case the player performs $O(\log k)$ algorithms covering the relevant degree range. And as before all relevant players include the correct guess in their range of guesses, implying that the union of messages of all players contain a triangle with high probability following the same analysis as in the non-oblivious case. 

The player splits the relevant range of degrees into two cases. For every degree guess, $d'$, where $\sqrt{n} \leq  d' \leq \frac{4k}{\eps}\bar{d^j}$ (Note that when $\bar{d^j} \leq \frac{\eps\sqrt{n}}{4k}$ this range is empty) the player simulates the algorithm for high degrees as we just described (with an edge limit of $O(\sqrt[3]{n\bar{d^j}}\log{n}\log{(k\log{n})})$ for each algorithm). 

If indeed $d \geq \sqrt{n}$ then correctness and complexity analysis for that case is the same as when $\bar{d^j} \geq \sqrt{n}$. 

Whereas for every guess, $d'$, where $\bar{d^j} \leq d' \leq \sqrt{n}$, the player simulates the $Alg_{Low}$ using $d'$ instead of $d$. More concretely, the player samples into $S$ each vertex with probability $p = \min\{\frac{c}{d'}, 1\}$, and the of sampling into $R$ each vertex with probability $\Theta(\frac{1}{\sqrt{n}})$ as in the original algorithm remains the same (the players can use the same $R$ across all simultaneous instances).

We use a cap of $O(\sqrt{n}\log{n}\log{(k\log{n})})$ edges for each instance of $Alg_{Low}$ and, as the we promptly prove, it suffices for the correct instance.

\begin{lemma}
A complexity bound of sending at most $O(\sqrt{n}\log{n}\log{(k\log{n})})$ edges for each instance of $Alg_{Low}$, suffices, for the case where $d \leq \sqrt{n}$, for the protocol pertaining to the correct guess to send all the edges in its corresponding subgraph.
\end{lemma}
\begin{proof}
The expected number of edges player $j$ has in $R$ is $\Theta(\frac{n\bar{d^j}}{n}) = \Theta(\bar{d^j}) = O(\sqrt{n})$. It is not surprising that the expected number did not increase, as the sample size does not depend on the average degree, and we have already assumed, in our previous analysis, the worst case of each edge in $R$ appearing in all inputs. 

For the correct guess, $d' = d \leq \sqrt{n}$, the expected number of edges player $j$ has connecting $S$ and $R$ is $\Theta(\frac{n\bar{d^j}}{d\sqrt{n}}) = O(\sqrt{n})$. Therefore, the expected number of edges player $j$ needs to send overall for the correct guess is $O(\sqrt{n})$, and indeed, for all guesses where $d' \leq \sqrt{n}$,
player $j$ can limit the edge budget of each algorithms with a bound of $O(\sqrt{n}\log{n}\log{(k\log{n})})$ following lemma \ref{lm:deviation_bound}, with all $k$ players not exceeding the bound with high probability. 

Since the complexity (and the edge cap) is higher for the simulations of $Alg_{Low}$ than for the simulations of $Alg_{High}$, the simulations of $Alg_{High}$ do not affect the overall computation of the complexity asymptotically.
\end{proof}

To conclude, when $d \leq \sqrt{n}$, the message cost of each player, $j$, is $$O(\max\{\sqrt{n}, (n\bar{d^j})^{1/3}\}\log^2{n}\log{k}\log{(k\log{n})}) = O((nd)^{\frac{1}{3}}\log^2{n}\log{k}\log{(k\log{n})})$$ thus the overall complexity of the protocol for all players is $O(k(nd)^{\frac{1}{3}}\log^2{n}\log{k}\log{(k\log{n})})$.

Whereas when $d \leq \sqrt{n}$ then $\bar{d^j} \leq \sqrt{n}$, and the complexity of player $j$ is $O(\sqrt{n}\log^2{n}\log^{k})$, thus for $k$ players we get $O(k\sqrt{n}\log^2{n}\log^{k})$.

We summarize our complete procedure for all cases in $FindTriangleSimOblivious(G)$. The correctness follows the fact that all relevant players participate in the instance pertaining to the correct guess, and what we have proved about the edge cap not limiting that instance. 

\begin{theorem}
The problem of $d$-oblivious triangle detection in the simultaneous model can be solved with communication cost of $O(k\sqrt{n}\log^2{n}\log{k}\log{(k\log{n})})$ for $d = O(\sqrt{n})$, and in $O(k(nd)^{\frac{1}{3}}\log^2{n}\log{k}\log{(k\log{n})})$ for $d = \Omega(\sqrt{n})$, by a single algorithm, with constant error in both cases.
\end{theorem}

\begin{algorithm}
	\caption{\small FindTriangleSimOblivious($G$)}
	\label{alg:sim_oblivious}
	\begin{algorithmic}[1]
		\State Each player $j \in [k]$ runs simultaneously $O(\log k)$ protocols - for each degree guess $d'$ - covering its relevant degree range, $D_j = [\bar{d^j},\frac{4k}{\eps}\bar{d^j}]$:
		\State for each guess $d' \geq \sqrt{n}$: 
		\State \quad run $Alg_{High}$ and send up to $O({(nd)}^(1/3)\log n \log{(k\log{n})})$ edges
		\State for each guess $d' < \sqrt{n}$:
		\State \quad run $Alg_{Low}$ and send up to $O(\sqrt{n}\log n \log{(k\log{n})})$ edges
		\State The Referee checks whether the union of edges it received contains a triangle, and outputs accordingly.
		
	\end{algorithmic}
\end{algorithm}

\section{Lower Bounds}
\label{section:lower}
% !TeX root = main.tex

Our main result in this section is the following:
\begin{theorem}\label{th:main_lower}
	For any $d = O(\sqrt{n})$, let $T^\eps_{n,d}$ be the task of finding a triangle edge in graphs of size $n$ and average degree $d$
	which are $\eps$-far from triangle-free.
	Then for sufficiently small constant error probability $\delta < 1/100$ we have:
	\begin{enumerate}[(1)]
		\item For $k > 3$ players: $\CC^{sim}_{k, \delta}(T_{n,d}^\eps) = \Omega\left( k \cdot {(nd)}^{1/6} \right)$.
		\item For $3$ players: $\CC^{sim}_{3, \delta}(T_{n,d}^\eps) = \Omega\left( {(nd)}^{1/3} \right)$.
	\end{enumerate}
	\label{thm:sim_sqrt}
\end{theorem}

\vspace{-0.5cm}

To show both results, we first prove them for average degree $d = \Theta(\sqrt{n})$, and then easily obtain the result for lower degrees by embedding a dense subgraph of degree $\Theta(\sqrt{n})$ in a larger graph
with lower overall average degree. 

To prove (1), 
we begin by proving that for graphs of average degree $\Theta(\sqrt{n})$, \emph{three} players require $\Omega(n^{1/4})$ bits of communication to solve $T^\eps_{n,\sqrt{n}}$
in the \emph{one-way} communication model, where Alice and Bob send messages to Charlie, and then Charlie outputs the answer. In fact, our lower bound is more general, and allows Alice and Bob to communicate back-and-forth for as many rounds as they like, with Charlie observing the transcript.
We then ``lift'' the result to $k > 3$ players communicating simultaneously, using symmetrization \cite{PVZ12}. 

To prove (2), we show directly that in the simultaneous communication model, three players require $\Omega(\sqrt{n})$ bits to solve $T^{\eps}_{n,d}$ in graphs of average degree $\Theta(\sqrt{n})$.

Our lower bounds actually bound the \emph{distributional} hardness of the problems: we show an input distribution $\mu$ on which any protocol that has a small probability of error \emph{on inputs drawn from $\mu$} requires high communication. 
This is stronger than worst-case hardness, which would only assert that any protocol that has small error probability \emph{on all inputs} requires high communication.

\subsection{Information Theory: Definitions and Basic Properties}
We start with an overview of our information theory toolkit, which is our primary technical apparatus for directly proving lower bound (as opposed to reductions, which we also use to derive subsequent results). 

\begin{definition}
	The \emph{mutual information} between two random variables is $\rv{\MI}(\rv{X} ; \rv{Y}) = H(\rv{X}) - H(\rv{X} | \rv{Y}) = \E_{y \sim Y} \left[ \Div{\mu(\rv{X}|\rv{Y}=y) }{ \mu(\rv{X})} \right]$.
\end{definition}

\begin{lemma}[Super-additivity of information]
	If $\rv{X}_1,\ldots,\rv{X}_n$ are independent, then
	\begin{equation*}
		\MI(\rv{X}_1, \ldots, \rv{X}_n ; \rv{Y}) \geq \sum_{i = 1}^n \MI(\rv{X}_i ; \rv{Y}).
	\end{equation*}
	
	\label{lemma:superadditivity}
\end{lemma}

\begin{lemma}
	Let $p, q \in (0,1)$, and let $\Div{q}{p}$ denote the KL divergence between $\mathrm{Bernoulli}(q)$ and $\mathrm{Bernoulli}(p)$.
	Then for any $p < 1/2$ we have $\Div{q}{p} \geq q-2 p$.
	\label{lemma:linear_Pinsker}
\end{lemma}
\begin{proof}
	Since the divergence is non-negative, it suffices to show that for $p < 1/2$, for any $q \geq 2 \cdot p$
	we have
	$\Div{q}{p} \geq q-2 p$.

	For convenience, let us write $q = p + x$, where $-p \leq x \leq 1-p$.
	Our goal is to show that when $x \geq p$, 
	we have $\Div{p+x}{p} \geq x - p$.

	Consider the difference
	\begin{align*}
		g(x,p) &= \Div{p+x}{p} - (x - p)
		\\
		&=
		(p+x) \log \frac{p+x}{p} + (1 - p - x) \log \frac{1 - p - x}{1 - p} - (x - p).
	\end{align*}
	At $x = p$ we have $g(x,p) \geq 0$ (since divergence is always non-negative); we will show that 
	the derivative w.r.t.\ $x$ is non-negative for $x \geq p$, and hence $g(x,p) \geq 0$ for any $x  \geq p$.

	Taking the derivative with respect to $x$, we obtain
	\begin{align*}
		\frac{\partial }{\partial x} g(x,p)
		&= \log \frac{p+x}{p} + (p+x) \cdot \frac{p}{p+x} \cdot \frac{1}{p \ln 2} - \log \frac{1-p-x}{1-p} - (1-p-x) \cdot \frac{1-p}{1-p-x} \cdot \frac{1}{(1-p) \ln 2} - 1
		\\
		&=
		\log \left( 1+ \frac{x}{p} \right) - \log\left( 1 - \frac{x}{1-p} \right) - 1
	\end{align*}
	The derivative is increasing in $x$, and since we consider only $x \geq p$, it is sufficient to show that it is non-negative at $x = p$:
	\begin{align*}
		\frac{\partial}{\partial x} g(x,p) \Big|_{x = p}
		&=
		\log \left( 1 + 1\right) - \log \left( 1 - \frac{p}{1-p} \right) - 1
		\\
		&\geq
		1 + \frac{p}{1-p} - 1 \geq 0.
	\end{align*}
	In the last step we used the fact that $\log (1 - z) \leq -z$ for any $z \in (0,1)$;
	in our case, since $p < 1/2$, we have $p/(1-p) < 1$.

\end{proof}

\subsection{Random Graph of Degree $\Theta(\sqrt{n})$}\label{sec:lower_sqrt_n}

In this section, we derive our main results, lower bounds for one-way and simultaneous communication, all using a single distribution, $\mu$, for graphs of average degree $\Theta(\sqrt{n})$, whose edges are shared among $3$ players. In the subsequent sections we move on to showcase methods to generalize these results for $k$ players and other average degrees. 

\subsubsection{The input distribution and its properties}

Our lower bounds for degree $\Theta(\sqrt{n})$ use the following input distribution, $\mu$:
we construct a tripartite graph $G = (U \cup V_1 \cup V_2, E)$, where each edge appears iid with probability $\gamma /\sqrt{n}$ for some constant $\gamma$.

This distribution has very high probability that the input is $\eps$-far from being triangle-free,
but it does not guarantee it with probability 1.
Still, if we can show some task (finding a triangle, or finding a triangle-edge) is hard on $\mu$,
then it is also hard on the distribution $\mu'$ obtained from $\mu$ by conditioning on the input being $\eps$-far from triangle-free.

\begin{observation}
	Let $\Pi$ be a protocol for some task $T$, with error probability at most $\delta$ on some distribution $\mu$ supported on a class $\mathcal{X}$ of inputs. Then for any $\mathcal{Y} \subseteq \mathcal{X}$, the error probability of $\Pi$ on $\mu|\mathcal{Y}$ is at most
	$\delta / \Pr_\mu\left[ \mathcal{Y} \right]$.
\end{observation}
\begin{proof}
	We can write:
\begin{align*}
	\delta \geq
	&\Pr\left[ \text{$\Pi$ errs on $\rv{X}$} \right]
	\\
	&
	=\Pr\left[ \text{$\Pi$ errs on $\rv{X}$} \given \rv{X} \in \mathcal{Y} \right] \Pr\left[ \rv{X} \in \mathcal{Y} \right]
	+
	\Pr\left[ \text{$\Pi$ errs on $\rv{X}$} \given \rv{X} \not \in \mathcal{Y} \right] \Pr\left[ \rv{X} \not \in \mathcal{Y} \right]
	\\
	&\geq
	\Pr\left[ \text{$\Pi$ errs on $\rv{X}$} \given \rv{X} \in \mathcal{Y} \right] \Pr\left[ \rv{X} \in \mathcal{Y} \right].
\end{align*}
The claim follows.
\end{proof}

In our case we have:
\begin{lemma}
	When $\gamma$ is sufficiently small, 
	a graph sampled from $\mu$ is $O(1)$-far from triangle-free with probability at least $1/2$.
	\label{lemma:mu_eps_far}
\end{lemma}
\begin{proof}
	Let $T$ be the random variable of the set of triangles in the graph, and let $I$ be the set of pairs of triangles that share an edge. let  
	\begin{align*}
	&\E[|T|] = \binom{n}{3}(\frac{\gamma}{\sqrt{n}})^3 \geq \frac{\gamma^3}{12}   n^{3/2}\\
	&\E[|I|] = 3\binom{n}{3}(n-3)(\frac{\gamma}{\sqrt{n}})^5 \leq \frac{1}{2}\E[|T|]
	\end{align*}
	Where the last inequality follows from choosing a sufficiently small $\gamma$.
	It follows that $\E[|T| - |I|] \geq \frac{1}{2}|T|$. Let $D$ denote the maximal size of a set disjoint triangles in the graph. Note that $D \geq |T|-|I|$, since given the set $T$ of triangles, we can for each pair in $I$ choose one of the intersecting triangles, and remove the other from $T$. this process halts after $|I|$ steps and we are left with a set disjoint triangles of size at least $|T|-|I|$. Therefore $\E[D] \geq \frac{\gamma^3}{24} \cdot n^{3/2}$.   Denote $X = n^{2/3} - D$, and let $|E|$ be the size of the set of edges in the graph. Trivially $|E| \geq D$, therefore  $$Pr(X \leq 0) \leq \Pr(n^{2/3}-|E| \leq 0) = \leq \Pr(n^{2/3} \leq |E|) \leq e^{-m^2/((1-\gamma)^2)}$$ where the next to last inequality follows from chernoff bound on the number of edges in the graph. Since $X$ gets negative values with exponentially small probability, and is only polynomial in value, it holds that $\E[X|X > 0] \leq (1+o(1))\E[X]$. Therefore
	
	For convenience denote $c_1 = \frac{\E[D]}{2n^{2/3}} = \frac{\gamma^3}{48}$. 	It follows that
	\begin{flalign*}
	&\Pr(D \leq c_1 n^{3/2}) = \Pr(X \geq (1-c)  n^{3/2}) = \\
	& \Pr(X \geq (1-c_1) n^{3/2}|X>0)\Pr(X>0) + \Pr(X \geq (1-c_1) \cdot n^{3/2}|X \leq 0)\Pr(X \leq 0) \leq \\
	&\Pr(X \geq (1-c_1) \cdot n^{3/2}|X>0) + e^{-m^2/((1-\gamma)^2)} \leq \frac{E[X|X > 0]}{(1-c_1)n^{2/3}}  \leq (1+o(1))\frac{E[X]}{(1-c_1)n^{2/3}}+o(1) = \\ &(1+o(1))\frac{n^{2/3}-E[D]}{(1-c_1)n^{2/3}}+o(1)  = (1+o(1))\frac{1-2c_1}{1-c_1} + o(1)
	\end{flalign*} 
	${(1-c_1)n^{2/3}}$ is a constant smaller than $1$, meaning $(1+o(1))\frac{1-2c_1}{1-c_1} + o(1)$ is smaller than some constant $c_2 < 1$. Therefore with constant probability there are at least $c_1$ disjoint triangles.

\end{proof}
Therefore, any lower bound we prove for $\mu$ translates to asymptotically the same bound on a distribution that is $\eps$-free from triangle-free, namely, $\mu$ conditioned on being $\eps$-free from triangle-freeness.

Let $\rv{X}_e$ be an indicator variable for the presence of edge $e$ in the input graph.
For a transcript $t$ of a communication protocol $\Pi$,
let
\begin{equation*}
	\Delta_t(e) \coloneq \Pr\left[ \rv{X}_e = 1 \given \rv{\Pi} = t \right] - 2\gamma/\sqrt{n}.
\end{equation*}

\begin{lemma}
	We have:
	\begin{equation*}
		\E_{t \sim \pi} \left[ \sum_e \Delta_t(e) \right] \leq |\Pi|.
	\end{equation*}
	\label{lemma:sum_delta}
\end{lemma}
\begin{proof}
For each edge $e$, the prior probability that $e \in \rv{E}$ is $\gamma/\sqrt{n}$, so by Lemma~\ref{lemma:linear_Pinsker},
for any transcript $t$,
\begin{equation*}
	\Delta_t(e) \leq \Div{\pi(\rv{X}_e|\rv{\Pi}=t)}{\pi(\rv{X}_e)},
\end{equation*}

By super-additivity of information,
\begin{align*}
	|\Pi|
	&\geq
	\MI( \rv{\Pi} ; \rv{E} )
	\geq
	\sum_e \MI( \rv{\Pi} ; \rv{X}_e )
	\\
	&=
	\E_{t \sim \pi}
	\left[
		\sum_e \Div{\pi(\rv{X}_e|\rv{\Pi}=t)}{\pi(\rv{X}_e)}
	\right]
	\\
	&\geq
	\E_{t \sim \pi}
	\left[
		\sum_e \Delta_t(e).
	\right]
\end{align*}
\end{proof}

\paragraph{Covered and reported edges.}
Our lower bounds show that it is hard for the players to find an edge that belongs to a triangle.
Intuitively, in order to output such an edge, the players need to identify some edge $\set{v_1, v_2}$ that
\begin{inparaenum}[(a)]
\item is in the input, and
\item closes a triangle together with some third vertex $u$; that is, for some $u$, the edges $\set{u, v_1}$ and $\set{u, v_2}$ are also in the input.
\end{inparaenum}

We formalize the notion of ``finding'' an edge satisfying some property using the posterior probability of the edge satisfying this property given the transcript.

\begin{definition}[Reported edges]
	Given a transcript $t$,
	let
	\begin{equation*}
		\Rep(t) = \set{ e \in \mathcal{E} \st \Pr\left[ e \in \rv{E} \given \rv{\Pi} = t \right] \geq 9/10}
	\end{equation*}
	be the set of edges whose posterior probabilities of being in the input increase to at least $9/10$ when transcript $t$ is sent.
	We call the edges in $\Rep(t)$ \emph{reported}.
	
\end{definition}
\begin{definition}[Covered edges]
	Given a transcript $t$,
	let
	\begin{equation*}
		\CS{t} = \set{ e \in V_1 \times V_2 \st \Pr\left[ \exists u \in U : (u,v_1) \in \rv{E}_1 \wedge (u,v_2) \in \rv{E}_2 \given \rv{\Pi}t \right] \geq 9/10}
	\end{equation*}
	be the set of edges in $V_1 \times V_2$ whose posterior probability of being covered by a vee rises to at least $9/10$ upon observing transcript $t$.
	We say that edges in $\CS{t}$ are \emph{covered} by Alice and Bob. Let $\Cov{e}$ be an indicator for the event that $e \in \CS{\rv{\Pi}}$.
\end{definition}

\subsubsection{One-Way Communication}\label{one-way}
Consider a protocol $\Pi$ between three players --- Alice, Bob and Charlie --- where Alice and Bob communicate back-and-forth for as many rounds as they want, with Charlie observing their transcript, and finally Charlie outputs an edge from his side of the graph. We claim that the total amount of communication exchanged by Alice and Bob must be $\Omega(n^{1/4})$.

The underlying intuition for our proof is that by the end of the protocol Charlie needs to be informed by Alice and bob of at least $\Omega(\sqrt{n})$ vertex pairs in $V_1 \times V_2$ being covered with high certainty by a vee in their input. This is due to the fact, that only a $\Theta(\frac{1}{\sqrt{n}})$-fraction of these pairs is expected to have an edge connecting them. We prove that the number of pairs Alice and Bob can on average inform Charlie of being covered is at most quadratically larger than their bit-budget, which implies that $\Omega(n^{1/4})$ bits are required for such communication, that succeeds with high probability. 

This result is somewhat surprising as the a priori probability of any edge in Charlie's input belonging to a triangle is already constant, and elevating only one of these probabilities to $1-delta$ suffices for solving the problem. This observation is equally valid for simultaneous communication. 

\begin{theorem}
	For any constant $\gamma \in (0,1)$, if $\Pi$ solves the triangle-edge-finding problem under $\mu$ with error $\delta \leq 1/100$, then
	$|\Pi| = \Omega(n^{1/4})$.
	\label{thm:one_way}
\end{theorem}
\begin{proof}

	Suppose for the sake of contradiction that there is a protocol $\Pi$ with communication $\alpha n^{1/4}$, where $\alpha$ satisfies:
\begin{equation*}
	(100\alpha^2 + 10\alpha)
	< (9/20)/\gamma,
\end{equation*}
	and error $\delta \leq 1/100$.

Say that transcript $t$ of $\Pi$ is \emph{good} if $|\CS{t}| \geq \sqrt{n}/(2\gamma)$.

\begin{lemma}
	$\Pr\left[ \text{$\rv{\Pi}$ is good} \right] \geq 1-20\delta$.
	\label{lemma:good_transcripts}
\end{lemma}
\begin{proof}
	If $t$ is not a good transcript, then because $\CS{t}$ is independent of $\rv{E}_3$,
	\begin{equation*}
		\E\left[ |\rv{E}_3 \cap \CS{t}| \right] \leq (\gamma/\sqrt{n}) \cdot (\sqrt{n}/(2\gamma)) = 1/10.
	\end{equation*}
	By Markov, $\Pr\left[ \rv{E}_3 \cap \CS{t} \neq \emptyset \right] \leq 1/2$.
	Whenever $\rv{E}_3 \cap \CS{t} = \emptyset$, Charlie must output an edge that is either not in his input ($\rv{E}_3$),
	or not covered by $t$;
	in the first case this is an error, and in the second case, the probability of an error is at least $1/10$, independent of $\rv{E}_3$ (it depends only on $\rv{E}_1, \rv{E}_2$,
	which are independent of $\rv{E}_3$, even given $\rv{\Pi}=t$).
	Therefore, conditioned on $\rv{E}_3 \cap \CS{t} = \emptyset$, the error probability is at least $1/10$;
	and overall, for any $t$ that is not good,
	\begin{equation*}
		\Pr\left[ \text{error} \given \rv{\Pi} = t \right] \geq \Pr\left[ \rv{E}_3 \cap \CS{t} = \emptyset \right] \cdot (1/10) \geq 1/20.
	\end{equation*}
	Since the total probability of error is bounded by $\delta$, we obtain
	\begin{align*}
		\delta &\geq \Pr\left[ \text{error} \right] = \sum_t \Pr\left[ \text{error} \given \rv{\Pi} = t \right] \Pr\left[ \rv{\Pi} = t \right]
		\\
		&\geq \sum_{\text{bad $t$}} \Pr\left[ \text{error} \given \rv{\Pi} = t \right] \Pr\left[ \rv{\Pi} = t \right]
		\\
		&\geq
		\sum_{\text{bad $t$}} (1/20) \cdot \Pr\left[ \rv{\Pi} = t \right] = \Pr\left[ \rv{\Pi} \text{ is bad} \right]/20.
	\end{align*}
	The claim follows.
\end{proof}

Next, say that $t$ is \emph{informative} if:
\begin{equation*}
	\sum_{e \in U\times V_1 \cup U \times V_2} \Delta_t(e) \geq 10\alpha n^{1/4}.
\end{equation*}
\begin{lemma}
	$\Pr\left[ \text{$\rv{\Pi}$ is informative} \right] \leq 1/10$.
	\label{lemma:uninformative_transcripts}
\end{lemma}
\begin{proof}
	By super-additivity,
	\begin{align*}
		\alpha n^{1/4} &= |\Pi|
		\geq \MI( \rv{\Pi} ; \rv{E}_1 \cup \rv{E}_2)
		\\
		&\geq
		\sum_{e \in U\times V_1 \cup U \times V_2} \MI( \rv{\Pi} ; \rv{X}_e )
		\\
		&=
		\E_{t \sim \rv{\Pi}}
		\left[
			\sum_{e \in U\times V_1 \cup U \times V_2} \Div{ \pi(\rv{X}_e |\rv{\Pi}=t) } {\pi(\rv{X}_e)}
		\right]
		\\
		&\geq
		\E_{t \sim \rv{\Pi}}
		\left[
		\sum_{e \in U\times V_1 \cup U \times V_2} \Delta_t(e)
		\right].
		\tag{By Lemma~\ref{lemma:linear_Pinsker}}
	\end{align*}
	The claim follows by Markov.
\end{proof}

\begin{corollary}
	There exists a transcript which is both \emph{good} and \emph{uninformative}.
\end{corollary}
\begin{proof}
	By union bound, the probability that a transcript is either not good or informative is at most $20\delta+1/10 < 1$.
\end{proof}

We will now show that such a transcript \emph{cannot} exist, as an uninformative transcript cannot cover enough edges to be good.

For any particular transcript $t$ of $\Pi$, the inputs of the three players remain independent given $\rv{\Pi} = t$.
Therefore, for any edge $(v_1, v_2) \in V_1 \times V_2$,
\begin{align*}
	&\Pr\left[ \exists u : (u,v_1) \in \rv{E}_1 \wedge (u,v_2) \in \rv{E}_2 \given \rv{\Pi} = t \right]
	\leq
	\sum_{u \in U} \Pr\left[ (u,v_1) \in \rv{E}_1 \wedge (u,v_2) \in \rv{E}_2 \given \rv{\Pi} = t  \right]
	\\
	&=
	\sum_{u \in U} \left( \Pr\left[ (u,v_1) \in \rv{E}_1 \right] \Pr\left[ (u,v_2) \in \rv{E}_2 \given \rv{\Pi} = t  \right] \right)
	\\
	&=
	\sum_{u \in U} \left( \left(\Delta_t(u,v_1) + 2\gamma/\sqrt{n} \right) \left( \Delta_t(u,v_2) + 2\gamma/\sqrt{n} \right) \right)
	\\
	&=
	\sum_{u \in U} \left( \Delta_t(u,v_1) \Delta_t(u,v_2) \right) + 2(\gamma/\sqrt{n})\sum_{u \in U} \left(\Delta_t(u,v_1) + \Delta_t(u,v_2) \right).
\end{align*}

Now let $t$ be a transcript that is \emph{good}, that is, $|\CS{t}| \geq \sqrt{n}/(2\gamma)$, and also \emph{uninformative}. Let $S(t) \subseteq \CS{t}$ be a set of $\sqrt{n}/(2\gamma)$ covered edges (chosen arbitrarily from $\CS{t}$),
and let $W_1(t) \subseteq V_1$ and $W_2(t) \subseteq V_2$ be the endpoints of the edges in $S$. Since each edge $(v_1, v_2) \in S(t)$ is covered in $t$,
\begin{equation*}
	\sum_{u \in U} \left( \Delta_t(u,v_1) \Delta_t(u,v_2) \right) + 2(\gamma/\sqrt{n})\sum_{u \in U} \left(\Delta_t(u,v_1) + \Delta_t(u,v_2) \right) \geq 9/10,
\end{equation*}
and together we have
\begin{align*}
	&\sum_{(v_1, v_2) \in S(t)}
	\sum_{u \in U} \left[ \left( \Delta_t(u,v_1) \Delta_t(u,v_2) \right) + 2(\gamma/\sqrt{n})\sum_{u \in U} \left(\Delta_t(u,v_1) + \Delta_t(u,v_2) \right)\right]\\
	&\geq (9/10)|S| = (9/20)\sqrt{n}/\gamma.
\end{align*}

On the other hand,
\begin{align*}
	&
	\sum_{(v_1, v_2) \in S(t)}
	\sum_{u \in U} \left[ \left( \Delta_t(u,v_1) \Delta_t(u,v_2) \right) + 2(\gamma/\sqrt{n})\sum_{u \in U} \left(\Delta_t(u,v_1) + \Delta_t(u,v_2) \right) \right]
	\\
	&\leq
	\sum_{u \in U} \left( \sum_{v_1 \in V_1} \Delta_t(u,v_1) \right) \left( \sum_{v_2 \in V_2} \Delta_t(u, v_2) \right)
	+
	2(\gamma/\sqrt{n}) \sum_{v_1 \in C_1(t)} \sum_{v_2 \in C_2(t)} \left(\Delta_t(u,v_1) + \Delta_t(u,v_2) \right) 
	\\
	&\leq 
	\left( \sum_{u \in U} \sum_{v_1 \in V_1} \Delta_t(u,v_1) \right) \left( \sum_{u \in U} \sum_{v_2 \in V_2} \Delta_t(u, v_2) \right)
	\\
	&\qquad\qquad
	+
	2(\gamma/\sqrt{n}) \cdot |S(t)| \cdot \left[ \left( \sum_{u \in U} \sum_{v_1 \in V_1} \Delta_t(u,v_1)  \right) +  \left( \sum_{u \in U} \sum_{v_2 \in V_2} \Delta_t(u,v_2)  \right) \right]
	\\
	&\leq \left( 10\alpha n^{1/4} \right)^2 + 
	2(\gamma/\sqrt{n}) \cdot \sqrt{n}/(2\gamma) \cdot 10\alpha n^{1/4}
	\\
	&\leq
	(100\alpha^2 + 10\alpha)\sqrt{n}
	.
\end{align*}
We therefore have
\begin{equation*}
	(100\alpha^2 + 10\alpha)\sqrt{n}
	\geq (9/20)\sqrt{n}/\gamma,
\end{equation*}
contradicting our assumption about $\alpha$.
\end{proof}

\paragraph{Streaming Lower Bounds}\label{sec:stream}
There is a known connection between communication complexity, specifically, one-way communication, and space complexity in the data-stream model. In this model the input arrives as an ordered sequence that must be accessed in order and can be read only once, while the space complexity is defined as the maximal size of the memory used at any given point of the computation. As demonstrated in \cite{Alon:1996:SCA:237814.237823}, there is a generic reduction which proves that lower bounds on the one-way communication complexity of a problem, are also lower bounds on the space-complexity of the same problem in the data-stream model. Consequently, we get a corresponding lower bound of $\Omega(n^{1/4})$ on the space complexity of detecting a triangle edge (with the input graph distribution identical to the one in our model) in the data-stream model. 

We present here a sketch of the proof, as the data-stream model is not the focus of this work; for more details on the relationship between lower bounds in the two models refer to \cite{Alon:1996:SCA:237814.237823,Gal:2010:LBS:2078516.2078518}. 

Assume to the contrary that there exists an algorithm, $\mathcal{A}$, that solves the triangle-edge detection with space complexity $o(n^{1/4})$ in the data-stream model. This implies a one-way 3-player protocol, $\Pi$, of complexity $o(n^{1/4})$, which implies a contradiction (our "extended" one-way model is even more powerful than the more standard one-way model used in this reduction, where Alice sends one message to Bob, who then sends one message to Charlie, who has to output the answer), proving our initial assumption to be false. More concretely, $\Pi$ entails Alice running $\mathcal{A}$ on the input, which is viewed as the beginning of the stream, then sending the content of the memory (which is limited by $o(n^{1/4})$ bits) to Bob, who continues the computation of $\mathcal{A}$ on his input, which is viewed as the continuation of the stream, and once again sends the content of the memory to Charlie, who concludes the computation of $\mathcal{A}$ on his input, the final segment of the stream. 

We can apply the same reduction to the extended one-way lower bounds we derive later in this chapter for a more general average degree $d = O(\sqrt{n})$.

\subsubsection{Simultaneous Communication}
\label{app:lower_sim}

For \emph{simultaneous} protocols, it is not enough to have some covered edge that also appears in Charlie's input: the referee needs to \emph{know} (or believe) that it is in Charlie's input  --- that is, with good probability, the edge the referee outputs has a large \emph{posterior} probability of being in Charlie's input, given Charlie's message.

Say that edge $e$ is \emph{reported} by a transcript $t$ if
	$\Pr_{\rv{E} \sim \mu|t}\left[ e \in \rv{E} \right] \geq 9/10.$
The goal of the players is to provide the referee with some edge that is \emph{covered} by Alice and Bob and also \emph{reported} by Charlie.

We show that the ``best'' strategy for the players is to choose a set $T \subseteq V_1 \times V_2$ of $\Theta(n)$ edges,
and have Alice and Bob try to cover edges from $T$ and Charlie report edges from $T$.
The crux of the lower bound is showing that to target a \emph{fixed} set of edges $T$, Alice and Bob must give up their quadratic advantage: whereas in for in our analysis of the one-way lower bound, the sum of the cover probabilities was bounded by the \emph{square} of the sum-increase of individual edge probabilities ($\sum_e \Delta_t(e)$), here we show that we can bound it linearly, yielding a lower bound of $\Omega(\sqrt{n})$ instead of $\Omega(n^{1/4})$. 

Fix a deterministic simultaneous protocol $\Pi$, where the messages sent by the three players are $\rv{M}_1, \rv{M}_2$ and $\rv{M}_3$, respectively.
Let $\Pi(m_1, m_2, m_3)$ denote the edge output by the referee upon receiving messages $m_1, m_2$ and $m_3$ from the three players.
We freely interchange the messages with the inputs to the respective players, since the protocol is deterministic; e.g., we write $\Pi(\rv{E}_1, \rv{E}_2, \rv{E}_3)$
to indicate the referee's output upon receiving the messages sent by the players on input $(\rv{E}_1, \rv{E}_2, \rv{E}_3)$.

Let $C = \alpha \sqrt{n}$ be the number of bits sent by each player, where $\alpha$ will be fixed later. Let $\delta$ denote the error of $\Pi$ on $\mu$. Our goal is to show that when $\gamma$ and $\delta$ are sufficiently small, we require $\alpha = \Omega(1)$, so the communication complexity of the protocol is $\Omega(n)$.

In a simultaneous protocol, the messages sent by the players are independent of each other given the input.
In our case, because the inputs are \emph{also} independent of each other, the messages are independent even without conditioning on a particular input.
We therefore abuse notation slightly by omitting parts of the transcript that are not relevant to the event at hand.
Specifically, we let $\Rep(m_i)$ denote the set of edges covered by a message $m_i$ of player $i$ (this is independent of the other players' messages),
and we let $\CS{m_1, m_2}$ denote the edges covered by messages $m_1, m_2$ of Alice and Bob, respectively 
(again, this is independent of Charlie's message).
We also sometimes write the player's \emph{input} instead of its \emph{message}; because the protocol is deterministic, the message is a function of the input.

In any simultaneous protocol, the goal of the players is to provide the referee with an edge in Charlie's input that is both reported by Charlie and covered by Alice and Bob:
\begin{lemma}
	The probability that there exists an edge that is both reported by Charlie and covered by Alice and Bob is at least $1 - 10\delta$.
	That is,
	\begin{equation*}
		\Pr\left[ \Rep(\rv{M}_3) \cap \CS{\rv{M}_1, \rv{M}_2} \neq \emptyset  \right] \geq 1 - 10\delta.
	\end{equation*}
	\label{lemma:err_covered_rep}
\end{lemma}
\begin{proof}
	If the referee outputs an edge that is both covered and reported, then of course there must \emph{exist} such an edge.
	Let us therefore bound the probability that the referee outputs an edge that is either not reported or not covered.
	Call a triplet $(m_1, m_2, m_3)$ of messages ``bad'' if $\Pi(m_1, m_2, m_3) = e$, where $e$ is either not reported ($e \not \in \Rep(m_3)$)
	or not covered ($e \not \in \CS{m_1, m_2})$.

	The protocol errs whenever it outputs an edge $e \in \mathcal{E}_3$ that is not in Charlie's input $\rv{E}_3$, 
	or an edge that does not form a triangle together with some node $u \in U$.
	If $e$ is not reported (in $m_3$), then $\Pr\left[ e \in \rv{E}_3 \given \rv{M}_3 = m_3  \right] < 9/10$,
	and if $e$ is not covered (in $m_1, m_2$), then $\Pr\left[ \exists u \in U : (u,v_1) \in \rv{E}_1 \wedge (u,v_2) \in \rv{E}_2 \given \rv{M}_1 = m_1, \rv{M}_2 = m_2 \right] < 9/10$.
	Therefore, each bad triplet of messages contributes at least $1/10$ to the error probability of the protocol.
	Together we have
	\begin{align*}
		\delta &\geq \Pr\left[ \text{$\Pi$ errs} \right] \geq \sum_{\text{bad $(m_1, m_2, m_3$)}} \Pr\left[ (\rv{M}_1, \rv{M}_2, \rv{M}_3) = (m_1, m_2, m_3) \right] \cdot (1/10)
		\\
		&
		= \Pr\left[ \text{ $(\rv{M}_1, \rv{M}_2, \rv{M}_3)$ are bad}  \right] / 10.
	\end{align*}
	The claim follows.

\end{proof}

By Lemma~\ref{lemma:err_covered_rep}, we see that the players' ``best strategy'' is to try to ``coordinate'' the edges reported by Charlie with the edges covered by Alice and Bob, so that the referee can find an edge in the intersection.
Indeed, as a corollary we obtain:
\begin{corollary}
	$\E\left[ \sum_{e \in \Rep(\rv{E}_3)} \Pr [ \Cov{e} ] \right] \geq 1 - 10\delta$.
	\label{cor:sum_cov}
\end{corollary}
\begin{proof}
	Fix $\Rep(\rv{E}_3) = R$.
	By union bound and the independence of the players' inputs,
	\begin{align*}
		&\Pr\left[ R \cap \CS{\rv{M}_1, \rv{M}_2} \neq \emptyset \given \Rep(\rv{E}_3) = R \right]
		\leq
		\sum_{e \in R} \Pr\left[ e \in \CS{\rv{M}_1, \rv{M}_2} \right]
		=
		\sum_{e \in R} \Pr\left[ \Cov{e} \right]
	\end{align*}
	Therefore,
	\begin{align*}
		&\E\left[ \sum_{e \in \Rep(\rv{E}_3)} \Pr [ \Cov{e} ] \right]
		\\
		&=
		\sum_R \left( \E\left[ \sum_{e \in \Rep(\rv{E}_3)} \Pr [ \Cov{e} ] \given \Rep(\rv{E}_3) = R \right] \Pr\left[ \Rep(\rv{E}_3) = R \right] \right)
		\\
		&\geq
		\sum_R \left( \left( 
		\sum_{e \in R} \Pr\left[ \Cov{e} \right]
		\right)
		\Pr\left[ \Rep(\rv{E}_3) = R \right]
		\right)
		\\
		&\geq
		\sum_R \left( \Pr\left[ R \cap \CS{\rv{M}_1, \rv{M}_2} \neq \emptyset \given \Rep(\rv{E}_3) = R \right] \Pr\left[ \Rep(\rv{E}_3) = R \right]
		\right)
		\geq 1 - 10\delta.
	\end{align*}
\end{proof}

\paragraph{Analyzing Charlie's messages.}
First, observe that Charlie (and the other players) cannot report too many edges, except with small probability.
Each reported edge is ``a little expensive'':
\begin{lemma}
	Let $m_i$ be a message sent by player $i$.
	Assume that $\gamma < 1/2$.
	If $e \in \Rep(m_i)$, then for sufficiently large $n$ we have $\Div{ \pi(\rv{X}_e \given \rv{M}_i = m_i) }{\pi(\rv{X}_e)} \geq 9\log n / 40$.
	\label{lemma:reported_div}
\end{lemma}
\begin{proof}
	Since $e \in \Rep(m_i)$, the posterior probability that $\rv{X}_e = 1$ is at least $9/10 > \gamma / \sqrt{n}$. 
	Because $\Div{p}{q}$ increases as $|p - q|$ increases,
	for sufficiently large $n$,
	\begin{align*}
		\Div{ \pi(\rv{X}_e \given \rv{M}_i = m_i) }{\pi(\rv{X}_e)}
		&\geq
		\Div{ 9/10 }{\gamma/\sqrt{n}}
		\\
		&=
		(9/10) \log \frac{9/10}{\gamma / \sqrt{n}} + (1/10) \log \frac{1/10}{1 - \gamma/\sqrt{n}}
		\\
		&=
		-H(1/10) + (9/10) \log \frac{\sqrt{n}}{\gamma} + (1/10) \log \frac{1}{1 - \gamma/\sqrt{n}}
		\\
		&\geq
		-1 + \frac{9/10}{2} \log n 
		\geq \frac{9}{40} \log n.
	\end{align*}
	We used the fact that $\gamma < 1/2$, so $(9/10)\log (1/\gamma) > 0$, and also
	that $1 - \gamma/\sqrt{n} < 1$, and hence $\log(1 / (1 - \gamma/\sqrt{n})) > 0$.
\end{proof}
It follows that with a budget of $C$ bits, Charlie can only report roughly $C$ edges (in fact, somewhat less) in expectation:
\begin{corollary}
\begin{equation*}
	\E\left[ |\Rep(\rv{E}_3)| \right] \leq \frac{40\alpha}{9\log n} \sqrt{n}
\end{equation*}
\label{cor:rep_size}
\end{corollary}
\begin{proof}
	By the super-additivity of information,
	\begin{align*}
		\alpha \sqrt{n} = |M_3| &
		\geq \MI( \rv{M}_3 ; \rv{E}_3)
		\geq
		\sum_{e \in \mathcal{E}_3} \MI( \rv{M}_3 ; \rv{X}_e)
		\\
		&=
		\E_{m_3 \sim \rv{M}_3} \left[ \sum_{e \in \mathcal{E}_3} \Div{ \pi(\rv{X}_e \given \rv{M}_3 = m_3) }{\pi(\rv{X}_e)}\right]
		\\
		&\geq
		\E_{m_3 \sim \rv{M}_3} \left[ \sum_{e \in \Rep(m_3)} \Div{ \pi(\rv{X}_e \given \rv{M}_3 = m_3) }{\pi(\rv{X}_e)}\right]
		\\
		&\geq
		\E_{m_3 \sim \rv{M}_3} \left[ |\Rep(m_3)| \cdot \frac{9}{40} \log n\right].
	\end{align*}
	The claim follows.
\end{proof}

As we said above, since the referee ``wants'' to output an edge that is both reported and covered,
the goal of the players should be to provide it with such an edge.
Let us rank the edges in $V_1 \times V_2$ according to the probability
that they are covered by Alice and Bob: we write $V_1 \times V_2 = \set{e_1, \ldots, e_{n^2}}$, where $i \leq j$ iff $\Pr\left[ \Cov{e_i} \right] \geq \Pr\left[ \Cov{e_j} \right]$, breaking ties arbitrarily.

Let $\Top(\rv{E}_3)$ denote the set of $|\Rep(\rv{E}_3)|$ highest-ranking edges in $\rv{E}_3$.
Clearly,
\begin{equation}
	\sum_{e \in \Rep(\rv{E}_3)} \Pr\left[ \Cov{e} \right] \leq \sum_{e \in \Top(\rv{E}_3)} \Pr\left[ \Cov{e} \right].
	\label{eq:rep_vs_top}
\end{equation}
That is, ``it is in Charlie's interest'' to report edges from $\Top(\rv{E}_3)$, as this maximizes the probability that some reported edge is also covered.

Let $T = \set{ e_1, \ldots, e_m}$ be the $m$ highest-ranking edges in $V_1 \times V_2$, where
\begin{equation*}
	m = \frac{9}{80\alpha}n.
\end{equation*}
For any integer $k \geq 1$ we have:
\begin{equation*}
	\sum_{e_1,\ldots,e_{k\cdot m}} \Pr\left[ \Cov{e} \right] \leq k \cdot \sum_{e \in T} \Pr\left[ \Cov{e} \right].
\end{equation*}

Therefore,
\begin{align*}
	&\E\left[ \sum_{e \in \Rep(\rv{E}_3)} \Pr [ \Cov{e} ] \right]
	\\
	&\leq
	\E\left[ \sum_{e \in \Top(\rv{E}_3)} \Pr [ \Cov{e} ] \right]
	\\
	&=
	\sum_{i = 1}^{\rceil \log (n^2 / m)\rceil } 
	\E\left[ \sum_{e \in \Top(\rv{E}_3)} \Pr [ \Cov{e} ] \Bigg| 2^i \cdot m \leq  |\Top(\rv{E}_3)| \leq 2^{i+1} \cdot m \right]
	\Pr\left[  2^i \cdot m \leq  |\Top(\rv{E}_3)| \leq 2^{i+1} \cdot m \right]
	\\
	&\leq
	\sum_{i = 1}^{\rceil \log (n^2 / m)\rceil } 
	\left[ \left( 2^{i+1} \cdot \sum_{e \in T} \Pr\left[ \Cov{e} \right] \right)
	\cdot
	\frac{ \E\left[ |\Top(\rv{E}_3)| \right]}{2^i \cdot m}
	\right]
	\\
	&\leq
	\log n \cdot 2 \cdot \frac{40\alpha}{9\log n} \frac{\sqrt{n}}{m} \cdot \sum_{e \in T} \Pr\left[ \Cov{e} \right]
	\\
	&=
	\frac{\sum_{e \in T} \Pr\left[ \Cov{e} \right]}{\sqrt{n}}.
	\numberthis
	\label{eq:sum_cover_lower}
\end{align*}

\paragraph{Analyzing the cover probabilities}

We show that it is not possible for the two other players to have:
\begin{equation*}
	\sum_{e \in T} \Pr\left[ \Cov{e} \right] \geq \beta \cdot \sqrt{n},
\end{equation*}
where $\beta$ is a constant whose value will be fixed later.

\paragraph{Notation.}
Let $V^H$ be the set of nodes in $V_1 \cup V_2$ whose degree in $T$ is at least $\sqrt{n}$,
and let $V^L$ be the remaining nodes in $V_1 \cup V_2$.
Also, let $V^a_i = V^a \cap V_i$, for $a \in \set{L,H}$ and $i \in \set{1,2}$.

Since $|T| \approx n$, we have $|V^H| \leq c \cdot \sqrt{n}$, where $c = \frac{9}{80\alpha}$.

Let $T_1 = V_1^L \times V_2 \cup V_1 \times V_2^H$ and let $T_2 = V_1 \times V_2^L \cup V_1^H \times V_2$.
For edges in $T_1$, their endpoints in $V_1$ all have low degree in $T_1$ (edges in $V_1^L \times V_2$ have degree at most $\sqrt{n}$
in $T$, and edges in $V_1 \times V_2^H$ also have low degree in $T_1$, since $|V^H| \leq c \cdot \sqrt{n}$).
We have $T = T_1 \cup T_2$, so it suffices to bound the sum of the cover probabilities in $T_1$ and 
the sum in $T_2$. (The union is not disjoint; e.g., edges in $V^L_1 \times V^L_2$ appear in both sets, so we may be over-counting).
Let $N_S(v)$ denote the nodes adjacent to node $v$ in $S \subseteq V_1 \times V_2$.

We let $\rv{M}_1, \rv{M}_2$ be random variables denoting the messages sent by the two players, respectively. Let $\mathcal{M}_1,\mathcal{M}_2$ be the set of all possible messages for each player (resp.).

\paragraph{Bounding the cover probabilities in $T$.}

For any pair of messages $m_1, m_2$,
if $e = (v_1, v_2) \in \CS{m_1, m_2}$, then by union bound,
\begin{align}
	&\sum_{u \in U} \Pr\left[ (u, v_1) \in \rv{E}_1 \wedge (u,v_2) \in \rv{E}_2 \given \rv{M}_1 = m_1, \rv{M}_2 = m_2 \right]
	\nonumber
	\\
	&\geq
	\Pr\left[ \exists u: (u, v_1) \in \rv{E}_1 \wedge (u, v_2) \in \rv{E}_2 \given \rv{M}_1 = m_1, \rv{M}_2 = m_2 \right] \geq 9/10.
	\label{eq:covered_implies}
\end{align}
Because the edges in $\rv{E}_1$ and $\rv{E}_2$ remain independent given $\rv{M}_1 = m_2, \rv{M}_2 = m_2$, and the messages are also independent of each other and of the other player's input,
for each $u \in U$,
\begin{align*}
	&\Pr\left[ (u, v_1) \in \rv{E}_1 \wedge (u,v_2) \in \rv{E}_2 \given \rv{M}_1 = m_1, \rv{M}_2 = m_2 \right]
	\\
	&
	=
	\Pr\left[ (u, v_1) \in \rv{E}_1 \given \rv{M}_1 = m_1, \rv{M}_2 = m_2\right]
	\cdot 
	\Pr\left[ (u,v_2) \in \rv{E}_2 \given \rv{M}_1 = m_1, \rv{M}_2 = m_2 \right]
\\
&=
	\Pr\left[ (u, v_1) \in \rv{E}_1 \given \rv{M}_1 = m_1\right]
	\cdot 
	\Pr\left[ (u,v_2) \in \rv{E}_2 \given \rv{M}_2 = m_2 \right]
.
\end{align*}

Consider first the edges in $T_1$.
Plugging the above into~\eqref{eq:covered_implies},
and also writing $\Pr\left[ (u,v_1) \in \rv{E}_1 \given \rv{M_1} = m_1 \right] = \Delta_{m_1}(u,v_1) + 2\gamma / \sqrt{ n}$ (where $\Delta_{m_1}$ is the L1 difference between the posterior and the prior),
we obtain:
\begin{align}
	\sum_{u \in U} \left[ 
		\left( \Delta_{m_1}(u,v_1) + 2\gamma / \sqrt{n} \right) 
\Pr\left[ (u,v_2) \in \rv{E}_2 \given \rv{M}_2 = m_2 \right]
	\right]
	\geq
	9/10.
	\label{eq:covered_implies1}
\end{align}
Multiplying both sides by $\Pr\left[ \rv{M}_2 = m_2 \right]$,
and summing across all $m_2$ such that $(v_1,v_2) \in \CS{m_1, m_2}$,
we get that for any $m_1$,
\begin{align*}
	&
	\sum_{m_2 : (v_1, v_2) \in \CS{m_1, m_2}} \sum_{u \in U} 
	\left[
	\left( \Delta_{m_1}(u,v_1) + 2\gamma/\sqrt{n}\right)
\cdot
\Pr\left[ (u,v_2) \in \rv{E}_2 \given \rv{M}_2 = m_2 \right]
\Pr\left[ \rv{M}_2 = m_2 \right]
\right]
	\\
	&
	=
	\left( \sum_{u \in U}  
	\left( \Delta_{m_1}(u,v_1) + 2\gamma/\sqrt{n}\right)
	\right)
	\cdot
	\left(
	\sum_{m_2 : (v_1, v_2) \in \CS{m_1, m_2}} 
\Pr\left[ (u,v_2) \in \rv{E}_2 \given \rv{M}_2 = m_2 \right]
\Pr\left[ \rv{M}_2 = m_2  \right]
\right)
	\\
	&
	\geq
	(9/10)
	\sum_{m_2 : (v_1, v_2) \in \CS{m_1, m_2}} 
\Pr\left[ \rv{M}_2 = m_2 \right]
\\
&=
(9/10) \Pr\left[ \Cov{v_1, v_2} \given \rv{M}_1 = m_1 \right].
\end{align*}

Notice that for any $u \in U$,
\begin{align*}
	&\sum_{m_2 : (v_1, v_2) \in \CS{m_1, m_2}}
	\Pr\left[ (u,v_2) \in \rv{E}_2 \given \rv{M}_2 = m_2 \right]
	\Pr\left[ \rv{M}_2 = m_2 \right]
	\\
	&\leq
	\sum_{m_2 \in \mathcal{M}_2}
	\Pr\left[ (u,v_2) \in \rv{E}_2 \given \rv{M}_2 = m_2 \right]
	\Pr\left[ \rv{M}_2 = m_2 \right]
	\\
	&=
	\Pr\left[ (u,v_2) \in \rv{E}_2 \right] = \gamma / \sqrt{ n}.
\end{align*}

Therefore,
\begin{align*}
	\left( \sum_{u \in U}  
	\left( \Delta_{m_1}(u,v_1) + 2\gamma/\sqrt{n}\right)
	\right)
	\cdot
	(\gamma/\sqrt{n})
	\geq
(9/10) \Pr\left[ \Cov{v_1, v_2} \given \rv{M}_1 = m_1 \right].
\end{align*}
Now, taking the expectation over all $m_1$,
\begin{align*}
	&\gamma/\sqrt{n} \E_{\rv{M}_1}\left[ \sum_{u \in U} (\Delta_{\rv{M}_1}(u,v_1) + 2\gamma/\sqrt{n} )  \right]
	\geq
	(9/10) \E_{\rv{M}_1} \left[ \Pr\left[ \Cov{v_1,v_2} \given \rv{M}_1 \right] \right]
	\\
	&
	=
	(9/10) \Pr\left[ \Cov{v_1, v_2} \right].
\end{align*}
Summing across all $v_2$ such that $(v_1, v_2) \in T_1$, and using the fact that the degree of $v_1$ in $T_1$
is at most $c \cdot \sqrt{n}$,
\begin{align*}
	&
	c\sqrt{n} \cdot \gamma/\sqrt{n} \E_{\rv{M}_1}\left[ \sum_{u \in U} (\Delta_{\rv{M}_1}(u,v_1) + 2\gamma/\sqrt{n} )  \right]
	\geq
	\sum_{v_2 \in N_{T_1}(v_1)} \left( \gamma/\sqrt{n} \E_{\rv{M}_1}\left[ \sum_{u \in U} (\Delta_{\rv{M}_1}(u,v_1) + 2\gamma/\sqrt{n} )  \right]\right)
	\\
	&
	\geq
	(9/10) \sum_{v_2 \in N_{T_1}(v_1)} \Pr\left[ \Cov{v_1, v_2} \right].
\end{align*}
And now, summing over all $v_1 \in V_1$,
\begin{align*}
	&c\gamma \sum_{v_1 \in V_1} \E_{\rv{M}_1}\left[ \sum_{v_1 \in V_1} \sum_{u \in U} \Delta_{\rv{M}_1}(u,v_1) + 2\gamma /\sqrt{ n} \right]
	\\
	&=
	c\gamma \left( \E_{\rv{M}_1}\left[ \sum_{v_1 \in V_1} \sum_{u \in U} \Delta_{\rv{M}_1}(u,v_1) \right] + 2\gamma \sqrt{ n} \right)
	\geq
	(9/10) \sum_{(v_1, v_2) \in T_1)} \Pr\left[ \Cov{v_1, v_2} \right].
\end{align*}
Using Lemma~\ref{lemma:sum_delta} we obtain:
\begin{equation*}
	\sum_{e \in T_1} \Pr\left[ \Cov{e} \right] \leq \frac{c\gamma(\alpha+2)}{9/10}\sqrt{n}.
\end{equation*}

For edges in $T_2$ the argument is symmetric.

Together we have:
\begin{equation*}
	\sum_{e \in T} \Pr\left[ \Cov{e} \right] \leq \frac{2c\gamma(\alpha+2)}{9/10}\sqrt{n} \leq \frac{1}{25}(\alpha+2)\sqrt{n},
\end{equation*}
assuming that $\gamma$ is a sufficiently small constant.

Combining this with~\eqref{eq:sum_cover_lower}, we see that we must have $\alpha \geq 23/25$.

\subsection{Lifting 3-player Lower Bounds to $k$ Players}

Using symmetrization~\cite{PVZ12}, we ``lift'' our lower bounds for a constant number of players to general $k$-player lower bounds (Symmetrization was developed in~\cite{PVZ12} to lift unrestricted 2-player lower bounds to unrestricted $k$-player lower bounds.) Interestingly, our symmetrization reduction transforms a \emph{simultaneous} $k$-player protocol into a \emph{one-way} 3-player (or 2-player) protocol, so in order to obtain lower bounds on simultaneous protocols for $k$ players we need to first prove lower bounds on one-way protocols for a small number of players. This curious behavior turns out to be inherent, at least for large $k$: a simultaneous protocol can emulate a one-way protocol, by having each player send their entire input to the referee with probability $1/k$, and otherwise send their message under the one-way protocol. The referee can, with constant probability, take the role of one of the players, whose input the referee received, and compute the answer using the messages from the other players. When $k$ is sufficiently large, this may be cheaper than the simultaneous protocol.

We say that a $k$-player distribution $\mu$ is \emph{symmetric} if the marginal distribution of each player's input is the same.

\begin{theorem}\label{th:k_players}
	Let $P$ be a graph property,
	Suppose that $\mu$ is a symmetric 3-player input distribution such that $\CC_{\mu,\delta}^{3,\rightarrow}(P^\eps) = C$.
	Then there is a $k$-player input distribution $\eta$ such that $\CC_{\eta,\delta}^{k,sim}(P^\eps) \geq (k/2)C$.
	\label{thm:sym}
\end{theorem}
\begin{proof}
	Let $\eta$ be the following distribution: we sample $(\rv{X}_1, \rv{X}_2, \rv{X}_3) \sim \mu$; we give $\rv{X}_1$ and $\rv{X}_2$ to two random players that are not player $k$,
	and the remaining players all receive $\rv{X}_3$.

	We show by reduction from the 3-player case that $\eta$ is hard for $k$ players. Let $\Pi$ be a simultaneous protocol for $k$ players that solves $P^\eps$ on $\eta$
	with error probability $\delta$.
	
	We construct a 3-player protocol $\Pi'$ as follows:
	Alice and Bob publicly choose two random IDs $i, j \in [k]$ ($i \neq j$), and take on the roles of players $i$ and $j$, using their actual inputs under $\mu$.
	Charlie will play the role of all the remaining players, using his input for each one of them, and also the role of the referee (who has no input).
	Let $\embed(i,j,X)$ denote the input thus constructed, where $X = (X_1, X_2, X_3)$.
	The resulting $k$-player input distribution is exactly $\eta$.

	To simulate the execution of $\Pi$, Alice and Bob simply send Charlie the messages players $i$ and $j$ would send under $\Pi$ to player $k$.
	Charlie computes the messages that each player $\ell \in [k ] \setminus \set{i,j}$ would send,
	and then, using these messages and the messages received from Alice and Bob,
	computes the output of the referee.

	The simulation adds no error --- on each input, it exactly computes the referee's output (or rather, it generates the correct distribution for the referee's output). Therefore,
	\begin{align*}
		\Pr_\mu\left[ \text{$\Pi'$ errs} \right]
		&=
		\sum_{X}
		\mu(X)
		\Pr\left[ \text{$\Pi'$ errs on $X$} \right]
		\\
		&=
		\frac{1}{k(k-1)} \sum_X \mu(X) \left[ \sum_{i,j} \Pr\left[ \text{$\Pi$ errs on $\embed(i,j,X)$} \right] \right]
		\\
		&=
		\sum_Y \eta(Y) \Pr\left[ \text{$\Pi$ errs on $Y$} \right]
		\leq \delta.
	\end{align*}

	What is the expected communication of $\Pi'$?
	Observe that since $\Pi$ is \emph{simultaneous}, each player's transcript is a (random) function of only its own input:
	in particular, 
	the distribution of player $i$'s transcripts is the same
	under any joint input distribution where player $i$'s input has the same marginal.
	In $\eta$, all players' inputs have the same marginal distribution --- the marginal distribution of each player's input in $\mu$.
	Therefore,
	\begin{align*}
		\E_{X \sim \mu} \left[ |\Pi'(X)| \right]
		&=
		\E_{i,j \sim U[k],X \sim \mu} \left[ |\Pi(\embed(i,j,X))| \right]
		\\
		&=
		\E_{i,j \sim U[k],X \sim \mu} \left[ |\Pi_i(\embed(i,j,X))| + |\Pi_j(\embed(i,j,X))| \right]
		\\
		&=
		\E_{i,j \sim U[k],Y \sim \eta} \left[ |\Pi_i(Y)| + |\Pi_j(Y)| \right]
		\\
		&=
		\E_{i \sim U[k],Y \sim \eta} \left[ |\Pi_i(Y)| \right]
		\\
		&=
		2\frac{1}{k} \sum_{i = 1}^k \E_{Y \sim \eta} \left[ |\Pi_i(Y)| \right]
		\\
		&=
		\frac{2}{k} \E_{Y \sim \eta} \left[ \sum_{i = 1}^k |\Pi_i(Y)| \right]
		=
		\frac{2}{k} \CC(\Pi).
	\end{align*}

\end{proof}

This result implies a $\Omega(k \cdot n^{1/4})$ lower bound for the problem of $k$ players trying to find a triangle-edge in a graph of average degree $d = \Theta(\sqrt{n})$ via simultaneous communication, 

For \emph{deterministic} and \emph{symmetric} protocols we can do a little better, by modifying the reduction: instead of constructing a \emph{one-way} protocol, we construct a \emph{simultaneous} protocol --- using the fact that the original $k$-player protocol is deterministic, Charlie can pick one of the players he simulates and send the message of only that one player to the referee, because we know that all the players simulated by Charlie will send the same message (as they receive the same input).

\subsection{Lower Bound for Degree $O(1)$}
\label{sec:lower_bhm}

For graphs with average degree $O(1)$, a lower bound was shown in the streaming model in \cite{Kallaugher17} reducing the \emph{Hidden Boolean Matching} problem, introduced in \cite{KR06} to triangle counting approximation in streaming. The same reduction yields a lower bound on triangle testing in two players one-way communication complexity. We present the reduction for the sake of completeness, and to show that it indeed holds in our model as well.  

We use the bound shown in \cite{VY11} but need only the bound for \emph{matchings} (rather than hypermatchings), we give here a simplified version of the problem;

\begin{definition}[Boolean Matching]
	In the \emph{Boolean Matching} problem, denoted $\prob{BM}_n$,
	Alice receives a vector $x \in \set{0,1}^{2n}$,
	and Bob receives a perfect matching $M$ on $2n$ vertices $\set{1,\ldots,2n}$ and a vector $w \in \set{0,1}^n$.
	We represent $M$ as an $n \times 2n$ matrix, where each row represents one edge of the matching:
	if the $i$-th edge of the matching is $\set{ j_1, j_2 } \subseteq [2n]^2$, then the $i$-th row of the matrix
	contains 1 in columns $j_1$ and $j_2$, and 0 elsewhere.

	The goal of the players is to distinguish the case where
	\begin{equation*}
		Mx \oplus w = \overrightarrow{0}
	\end{equation*}
	from the case where
	\begin{equation*}
		Mx \oplus w = \overrightarrow{1}.
	\end{equation*}

\end{definition}

\begin{theorem}
	The randomized one-way communication complexity of testing triangle-freeness in graphs with average degree $O(1)$ is $\Omega(\sqrt{n})$.
	\label{thm:one_way_const}
\end{theorem}
\begin{proof}
	Given inputs $X$ for Alice and $M,w$ for Bob, the players construct the following graph $G = (V, E)$, where
	$V = \set{ u } \cup \left( [n] \times [2] \right)$:
	\begin{itemize}
		\item For each bit $i \in [n]$ where $x_i = 0$, Alice adds the edge $\set{u, (i, 0)}$;
			for each bit $i \in [n]$ where $x_i = 1$, she adds the edges $\set{u, (i, 1)}$.
		\item For each edge $e_j = \set{ j_1, j_2}$ in his matching,
			\begin{itemize}
				\item If $w_j = 0$, Bob adds edges $\set{ (j_1, 0), (j_2,0)}$ and $\set{ (j_1, 1), (j_2, 1)}$;
				\item If $w_j = 1$, Bob adds edges $\set{ (j_1, 0), (j_2, 1)}$ and $\set{ (j_1, 1), (j_2, 0)}$.
			\end{itemize}
	\end{itemize}

	For each $j \in [n]$,
	let $M_j = \set{j_1, j_2}$.

	A triangle appears in the subgraph induced by vertices $\set{ u, (j_1,0), (j_1, 1), (j_2, 0), (j_2, 1)}$
	iff either $w_j = 0$ and $x_{j_1} = x_{j_2}$, or $w_j = 1$ and $x_{j_1} \neq x_{j_2}$.
	In other words, a triangle appears iff  $(M x \oplus w)_j = 0$.
	No other triangles appear in the graph.
	Therefore, if $M x \oplus w = \overrightarrow{0}$ then $G$ contains $n$ edge-disjoint triangles,
	and if $M x \oplus w = \overrightarrow{1}$ then $G$ is triangle-free.
	In the first case, $G$ is $1$-far from triangle-freeness.
\end{proof}

\subsection{Other Degrees}

We now show how we can extend a lower bound for a given average degree, $d$, to any lower degree, $d'$, by embedding dense inputs of degree $d$ into sparse graphs such that the average degree evens out to be $d'$. 

\begin{lemma}\label{lm:embed}
Let $d = \Theta(n^c)$ denote the average degree of the graph, and let $\CC(T^{\eps, d, n}) = \Theta(f(n))$ denote the communication complexity as a function of $n$, the number of vertices. Then for any  $d' \leq d$, we have $\CC(T^{\eps, \Theta(d'), n}) = \Theta(f({(d'n)}^{\frac{1}{1+c}}))$.
	\label{lemma:lower_degrees}
\end{lemma}
\begin{proof}
For graphs with $n' = {(d'n)}^{\frac{1}{1+c}}$ vertices and average degree $\Theta({(n')}^c)$, the communication complexity is $\Theta(f(n'))$.  We examine the following subset of graphs with $n$ vertices and average degree $d'$: any such graph, $G$, is a union of $(n-n')$ isolated nodes, and a graph, $G'$, which is either triangle-free or $\eps$-far from being triangle-free, and has $n'$ vertices and average degree ${(n')}^c$. The average degree of $G$ is $\Theta({(n')}^c) \cdot \frac{n'}{n} = \Theta(d')$, and its distance to being triangle-free is identical to that of $G'$, as it has no edges outside of $G'$. Since any triangle in $G$ must be contained in $G'$, solving the problem on $G$ is equivalent to solving it on $G'$. And since we asserted that the complexity of the problem for graphs with the stated properties of $G$ is $\Theta(f(n')) = \Theta(f{(d'n)}^{\frac{1}{1+c}}$, it is also the complexity of the problem for graphs of average degree $\Theta(d')$.
% Let $\mu(n)$ denote a maximally hard distribution for the problem of graphs with $n$ vertices and average degree $d = n^c$. The distributional complexity of the problem for this distribution is $\Theta(f(n))$, as it is maximally hard. We construct a new hard distribution $\mu'(n)$ for graphs of degree $d'$. We choose $m={(d'n)}^{\frac{1}{1+c}}$ vertices and use $\mu(m)$ to generate a subgraph, $G'$, on them. The rest of the $(n-m)$ vertices are isolated. All the edges of the resulting graph, $G$, are contained in $G'$, and hence it is also $\eps$-far from being triangle-free. Moreover, the average degree in $G'$ is $m^c$, implying the total number of edges in $G$ to be $m\cdot m^c = nd'$, which means the average degree is $d'$. Since there are no edges outside of $G'$, it is equivalent to solve the problem for $G'$ instead of $G$. And since $G'$ is created by $\mu(m)$, the distributional complexity of the problem is $\Theta(f(m)) = \Theta(f({(d'n)}^{\frac{1}{1+c}}))$. 
\end{proof}

Note that lemma \ref{lm:embed} holds regardless of the model of communication. Therefore, as a corollary, we can generalize the lower bounds we derived directly for graphs of average degree $\sqrt{n}$ to $d=O(\sqrt{n})$ (for 3 players in both cases). Specifically, the $\Omega(n^{1/4})$ bound for one-way communication and the $\Omega(\sqrt{n})$ bound for simultaneous communication extend to $\Omega({(nd)}^{1/6})$ and $\Omega({(nd)}^{1/3})$, respectively. Furthermore, lemma \ref{lm:embed}, combined with theorem \ref{th:k_players} and the lower bounds we proved in section \ref{sec:lower_sqrt_n} imply Theorem \ref{th:main_lower}, the main result of this section. 

\subsection{Discussion: Lower Bounds on the Communication Complexity of Property-Testing}
\label{sec:discuss}

Lower bounds on the ``canonical'' problems in communication complexity, such as Set Disjointness and Gap Hamming Distance~\cite{CR11}, cannot be leveraged to obtain property-testing lower bounds, at least for triangle-freeness. Some classical problems do feature a gap, where we are only interested in distinguishing two cases that are ``far'' from each other; however, for property-testing lower bounds, the gap needs to be around \emph{zero} (either we have no triangles, or we have many edge-disjoint triangles), while existing gap problems typically become easy unless the gap is centered \emph{far} from zero (for example, in Gap Hamming Distance, the players get vectors $x,y \in \set{0,1}^n$, and they need to determine whether their Hamming distance is greater than $n/2+\sqrt{n}$ or smaller than $n/2 - \sqrt{n}$).
In addition, because triangles are not ``independent'' of each other (if they share an edge), the \emph{direct sum} approach to proving lower bounds, which works well when we can break the problem up into many independent pieces, does not apply here.

\section{Summary}
\label{section:summary}
% !TeX root = thesis.tex

In this work we showed that in the setting of communication complexity, property testing can be significantly easier than exactly testing if the input satisfies the property: exactly determining whether the input graph contains a triangle was shown to require $\Omega(knd)$ bits in~\cite{Woodruff13}, but we showed that weakening the requirement to property-testing improves the complexity, and even simultaneous protocols can do better than the best exact algorithm with unrestricted communication.
However, the problem does not appear to become trivial, as shown by our lower bounds for simultaneous and restricted one-way protocols. Table \ref{tbl:summary} summarizes our main results.

\begin{table}[h]
	\centering
	\begin{tabular}{|c||c|c|c|}
		\hline
		& $d = \Theta(1)$ & $d = O(\sqrt{n})$ & $d = \Omega(\sqrt{n})$ \\ \hline
		
        \begin{tabular}[x]{@{}c@{}}$\bigtriangleup$-freeness \\  Unrestricted Communication \\ Upper bound \\\end{tabular} & \multicolumn{3}{c|}{$\tilde{O}(k\sqrt[4]{nd}+k^2)$} \\ \hline
		
        \begin{tabular}[x]{@{}c@{}}$\bigtriangleup$-freeness  \\ Simultaneous Communication \\ Upper bound \end{tabular} & \multicolumn{2}{c|}{$\tilde{O}( k\sqrt{n} )$} & $\tilde{O}( k\sqrt[3]{nd} )$ \\ \hline
        
		 \begin{tabular}[x]{@{}c@{}}$\bigtriangleup$-edge detection  \\ "Extended" One-Way Communication \\ 3 players \\ Lower bound \end{tabular} & \multicolumn{2}{c|}{$\Omega( \sqrt[6]{nd} )$} & | \\ \hline
        
		\begin{tabular}[x]{@{}c@{}}$\bigtriangleup$-edge detection  \\ Simultaneous Communication \\ 3 players \\ Lower bound \end{tabular} & \multicolumn{2}{c|}{$\Omega\left( \sqrt[3]{nd} \right)$} & | \\ \hline
    
   \begin{tabular}[x]{@{}c@{}}$\bigtriangleup$-edge detection  \\ Simultaneous Communication \\ Lower bound \end{tabular} & \multicolumn{2}{c|}{$\Omega \left( k \cdot \sqrt[6]{nd} \right)$} & | \\ \hline
     
  \begin{tabular}[x]{@{}c@{}}$\bigtriangleup$-freeness \\ Simultaneous Communication \\ Lower bound \end{tabular} & $\Omega(\sqrt{n})$ & \multicolumn{2}{c|}{|} \\ \hline    
	\end{tabular}
	\caption{Results summary;}
	\label{tbl:summary}
\end{table}

We have provided non-trivial upper bounds for the entire relevant degree range for both simultaneous and unrestricted communication. Our solutions have several desirable qualities. First, they can overcome the obstacle of not knowing the average degree in advance. Additionally, the algorithms solve not only the problem of triangle-freeness, but more specifically, the problem of triangle detection, which can only be harder. Finally, all solutions have a one-sided error - a graph is found to contain triangles only if a triangle is detected with probability $1$. We also address other variants and relaxations, such as the case where all inputs are disjoint, or a case where the players communicate via a blackboard visible to everyone, and describe how these guarantees can improve the complexity. In terms of more general contributions, we describe how to efficiently implement typical building blocks used in standard property-testing solutions, of which the most notable is the proposed procedure for approximating a vertex degree up to a constant, which can be used to solve the more general problem of approximating the number of distinct elements in a set. 

The task of proving non-trivial lower bounds for triangle-freeness is considerably harder. We have discussed the shortcomings of mainstream techniques in communication complexity for tackling this problem. Nevertheless, we have been able to produce a a tight lower-bound for the closely related problem of triangle-edge detection for $d = \Theta(\sqrt{n})$ in the simultaneous model. We have also been able to prove a lower-bound for an extended variation of one-way communication, which enabled us to derive a bound for $k$ players. Moreover, we showed how to extend these bounds, by a rather generic procedure, to lower average degrees. Finally, we demonstrated how to translate our one-way bounds  into streaming-lower bounds, once again via a generic (and well known) reduction. 

We believe that extending the lower bounds to protocols with unrestricted rounds, and strengthening them to apply to testing triangle-freeness rather than finding a triangle edge, will require techniques from Fourier analysis, like the ones used in~\cite{KR06} to show the lower bound on Boolean Hidden Matching (from which we reduce in Section \ref{sec:lower_bhm}). In addition, we believe that devising a hard distribution for dense graphs of degree $d = \omega(\sqrt{n})$, with desirable properties for proving lower bounds, will require some sophisticated utilization of Behrend graphs \cite{Alon:2006:TTG:1109557.1109589}. Finally, a worthwhile topic for related future research could be generalizing our techniques for detecting a wider class of subgraphs or testing other properties, relying on the property-testing relaxation. As demonstrated by this work, this relaxation can significantly reduce the complexity of an otherwise maximally hard problem, but not to a degree that it becomes trivial and uninteresting, as suggested by our lower bounds. More generally, there is much room for a more elaborate investigation of the interrelation between the models of communication complexity and property-testing, as alongside innate distinctions there seem to exist non-trivial similarities the extent of which is yet to be determined.

\section*{Acknowledgements}
% !TeX root = thesis.tex
We thank Noga Alon, Eldar Fischer and Dana Ron for fruitful discussions.

\bibliographystyle{plain}
\bibliography{bibliography}

\end{document}